\documentclass[a4paper,11pt]{article}

\usepackage{fullpage}
\usepackage{microtype}
\usepackage{amssymb}
\usepackage{comment}
\usepackage{amsmath, mathtools}
\usepackage{amsthm}
\usepackage{setspace} 
\usepackage{bm}
\usepackage{times}
\usepackage{fancybox}
\usepackage{color}
\usepackage{latexsym}
\usepackage{verbatim}
\usepackage{hyperref}
\usepackage{braket}
\usepackage{tikz}
\usepackage{thm-restate}
\usepackage{cleveref}
\usepackage{thmtools}
\newtheorem{theorem}{Theorem}[section]
\newtheorem{lemma}{Lemma}[section]

\newtheorem{claim}{Claim}[section]

\newtheorem{definition}{Definition}[section]
\newtheorem{problem}{Problem}[section]
\newtheorem{remark}{Remark}[section]

\newcommand{\supp}{\mbox{\small\rm Supp}}

\newcommand{\VP}{\mbox{\small\rm VP}} 
\newcommand{\VNP}{\mbox{\small\rm VNP}} 

\newcommand{\red}{\mbox{\small\rm RED}}

\newcommand{\Tau}{\mathrm{T}}

\usepackage{cleveref}
\usepackage{thm-restate}
\usepackage{hyperref}

\usepackage{cleveref}

\newcommand{\mon}{\mbox{\small\rm mon}} 
\newcommand{\ST}{\mbox{\small\rm ST}}

\newcommand{\E}{\mathbb{E}}

\newcommand{\R}{\mathbb{R}}

\renewcommand{\angle}[1]{{\langle} #1 {\rangle}}

\renewcommand{\mod}[1]{\ (\mathrm{mod}\ #1)}

\DeclareMathOperator{\var}{\mbox{\small\rm var}} 
\newcommand{\IP}{\mbox{\small\rm IP}}

\newcommand{\comst}{\mbox{\small\rm Comm-ST}} 
\newcommand{\disp}{\mbox{\small\rm Disc}}

\newcommand{\cproof}{\noindent{\it Proof of Claim}}

\newif\ifcomment

\commenttrue 
\DeclareMathOperator{\Good}{\mbox{\small\rm Good}} 
 
\DeclareMathOperator{\Bad}{\mbox{\small\rm Bad}}

\DeclareMathOperator{\poly}{\mbox{\small\rm poly}}

\newenvironment{appendix-lemma}[1]{\vspace{0.1in}\noindent{\bf Lemma~#1~} \em }{\vspace{0.1in}}
\newenvironment{appendix-theorem}[1]{\vspace{0.1in}\noindent{\bf Theorem~#1~} \em }{\vspace{0.1in}}

\title{Monotone Complexity of Spanning Tree Polynomial Re-visited} 
\author{Arkadev Chattopadhyay\thanks{TIFR, Mumbai. Partially supported by a MATRICS grant of the Science and Engineering Research Board, DST, India. {\tt arkadev.c@tifr.res.in}}
\and Rajit Datta\thanks{Goldman-Sachs, Bangalore. The work was done when the author was a research scholar at CMI partially supported by a TCS Fellowship. {\tt rajitcmi@gmail.com}}
\and Utsab Ghosal\thanks{CMI,Chennai {\tt ghosal@cmi.ac.in}}
\and Partha Mukhopadhyay\thanks{CMI, Chennai. {\tt partham@cmi.ac.in}}
}
\begin{document}
\maketitle
\begin{abstract}
We prove two results that shed new light on the monotone complexity of the spanning tree polynomial, a classic polynomial in algebraic complexity and beyond. 

First, 
we show that the spanning tree polynomials having $n$ variables and defined
over constant-degree expander graphs, have monotone arithmetic complexity $2^{\Omega(n)}$. This yields the first strongly exponential lower bound on monotone arithmetic circuit complexity for a polynomial in VP. 
Before this result, strongly exponential size monotone lower bounds were known only for explicit polynomials in VNP \cite{GS12, RY11, Sri20, CKR20, HY21}. 

Recently, Hrube\v{s} \cite{H20} initiated a program to prove lower bounds against \emph{general} arithmetic circuits by proving $\epsilon$-sensitive lower bounds for monotone arithmetic circuits for a specific range of values for $\epsilon\in (0,1)$. The first $\epsilon$-sensitive lower bound was just proved for a family of polynomials inside $\VNP$ by Chattopadhyay, Datta and Mukhopadhyay \cite{CDM21}. We consider the spanning tree polynomial $\ST_n$ defined over the complete graph of $n$ vertices and show that the polynomials 
 $F_{n-1,n} - \epsilon\cdot \ST_{n}$ and $F_{n-1,n} + \epsilon\cdot \ST_{n}$, defined over $n^2$ variables, have monotone circuit complexity $2^{\Omega(n)}$ if $\epsilon \geq 2^{- \Omega(n)}$ and $F_{n-1,n} \coloneqq \prod_{i=2}^n \big(x_{i,1} +\cdots+x_{i,n}\big)$ is the complete set-multilinear polynomial. 
 This provides the first $\epsilon$-sensitive exponential lower bound for a family of polynomials inside $\VP$. En-route, we consider a problem in 2-party, \emph{best partition} communication complexity of deciding whether two sets of oriented edges distributed among Alice and Bob form a spanning tree or not. We prove that there exists a \emph{fixed} distribution, under which the problem has low discrepancy with respect to every nearly-balanced partition. This result could be of interest beyond algebraic complexity.   
 
 Our two results, thus, are incomparable generalizations of the well known result by Jerrum and Snir \cite{JS82} which showed that the spanning tree polynomial, defined over complete graphs with $n$ vertices (so the number of variables is $n^2$), has monotone complexity $2^{\Omega(n)}$. In particular, the first result is an optimal lower bound and the second result can be thought of as a robust version of the earlier monotone lower bound for the spanning tree polynomial.  
 
\end{abstract}
\section{Introduction}\label{sec:intro}
A central problem in complexity theory aims to understand the following question: How much more powerful are computations with access to negations than computations that are monotone? While strong limitations have been known since almost three decades on monotone computations, in both Boolean and algebraic complexity, several basic gaps in our knowledge remain. We consider two such natural problems in the context of algebraic complexity.

The first question is as follows: Does there exist a monotone polynomial that can be efficiently computed by general circuits using cancellations but for which monotone circuits cannot do anything better than computing it by brute-force, i.e. monomial by monomial. The first big progress on this problem was made by Valiant \cite{Val79} in a seminal work. Soon after, Jerrum and Snir \cite{JS82} proved a similar result using a different polynomial called the spanning tree polynomial. In both these works, the monotone polynomials were respectively exhibited on $n$ variables, shown to be computed by general arithmetic circuits of polynomial in $n$ size, and were also shown to require monotone circuits of size $2^{\Omega(\sqrt{n})}$. While these results show that exponential savings in circuit size is possible by allowing negations, they leave the following question open: can the savings be pushed to the extremes, i.e. for any such polynomial can one prove a strongly exponential lower bound of $2^{\Omega(n)}$ on their monotone complexity. 

Strongly exponential lower bounds on monotone arithmetic circuit size have been also proven before, spanning efforts from the eighties to very recent times: first in the eighties by Kasim-Zade \cite{K86}, Gashkov \cite{G87}, Gashkov and Sergeev \cite{GS12} which remained relatively unknown. Then, Raz and Yehudayoff \cite{RY11} showed such bounds using sophisticated exponential sum estimates, more recently by Srinivasan \cite{Srinivasan19}, much more recently by Cavalar, Kumar and Rossman \cite{CKR20} using polynomials derived from error correcting codes. Just this year, Hrube\v{s} and Yehudayoff \cite{HY21} established such bounds using polyhedral combinatorics. Remarkably, in each case, the target polynomials are not known, and perhaps not expected, to lie in VP, the class of polynomials that have efficient general circuits. This left the following possibility open in principle: proving a strongly exponential lower bound of $2^{\Omega(n)}$ on the monotone complexity of a monotone polynomial was enough to show that it did not admit general circuits of polynomial size. 

Our first result finally rules out this possibility with two interesting features: first the polynomial used is a spanning tree polynomial, something that has been studied since at least Jerrum and Snir's classical work. Second, our argument is quite short.

Our second result concerns establishing 'robust' monotone lower bounds: consider an 'easy' polynomial like the full set multi-linear polynomial $F_{n,m} \coloneqq \prod_{i=1}^n \sum_{j=1}^m x_{i,j}$. A simple depth-2 monotone circuit of size $m+1$ (with unbounded fan-in) computes $F_{n,m}$. Can its complexity be significantly increased by a very small perturbation using another polynomial? Let $f$ be a set-multilinear monotone polynomial and $f_{\epsilon} \coloneqq F_{n,m} - \epsilon\cdot f$, for a small $\epsilon$. Obviously, for $\epsilon = 0$, the monotone complexity of $f_{\epsilon}$ is very low. Surprisingly, Hrube\v{s} \cite{H20} showed that if $f$ can be computed in polynomial size (say $s$) by general set-multilinear circuits, then $f_{\epsilon}$ has polynomial-size  monotone circuits for roughly doubly exponentially small in $s$, but non-zero,  $\epsilon$. This gives rise to a natural program of proving super-polynomial monotone lower bounds for such $f_{\epsilon}$, when $\epsilon$ is small as an approach to ultimately obtain lower bounds for general circuits. Hrube\v{s} calls this $\epsilon$-sensitive lower bounds for $f$. While Hrube\v{s} posed the challenge of establishing such lower bounds for the Permanent polynomial, none seemed to have been known for any explicit polynomial. In a very recent work, Chattopadhyay, Datta and Mukhopadhyay \cite{CDM21}, established the first such lower bound for a polynomial $f$ that lies in VNP but is not known to be in VP. More precisely, they could prove exponential monotone lower bounds for $f_{\epsilon}$ as long as $\epsilon \ge 2^{-\Omega(\sqrt{n})}$.

A natural question that arises from the work of \cite{CDM21} is as follows: how small an $\epsilon$ is needed to establish general circuit lower bounds from $\epsilon$-sensitive super-polynomial monotone ones? Hrube\v{s}' argument yields $\epsilon$ to be as small as doubly exponentially small to rule out $s$-size general circuits. Is that necessary? Or the kind of $\epsilon$ that \cite{CDM21} handles is (nearly) sufficient? We give a strong evidence against the latter by proving the first $\epsilon$-sensitive monotone lower bounds against a polynomial with $n$ variables that is in VP, for $\epsilon \ge 2^{-\Omega(\sqrt{n})}$. The polynomial we use is the spanning tree polynomial each of whose variables correspond to an edge of a complete graph on $n$ vertices. Our argument exploits the recently found connection by Chattopadhyay et. al. \cite{CDM21}, between the notion of rectangular discrepancy in communication complexity and that of $\epsilon$-sensitive monotone lower bounds. The main contribution of this result is to exhibit a (well known) polynomial in VP for which discrepancy based techniques still work. 

Our two results about the spanning tree polynomial can be naturally interpreted to give the following generalizations of the classical results of Valiant \cite{Val79} and Jerrum and Snir \cite{JS82}: the first result shows that VP contains polynomials that are maximally hard for monotone circuits. The second result shows that VP also contains polynomials that are quite robustly hard for monotone circuits. The work also opens up interesting avenues for further investigations that we describe in the last section. 
\subsection*{Details of Our Results}
Our first result about strongly exponential lower bounds on monotone complexity can be thought of as a 
reinterpretation of the argument of \cite{JS82} in more modern terms combined with the use of expander graphs. The idea of using expander graphs seems natural since expander graphs \emph{approximate} complete graph in intuitive sense of pseudo-randomness. Partly, this is also inspired by the result of Srinivasan \cite{Sri20}. Now, we explain our result in detail.  

Let $G$ be an undirected graph on $n$ vertices and let $\widetilde{G}$ be the directed graph obtained from $G$ which has edges $(u,v)$ and $(v,u)$ (in both directions) for every undirected edge $(u,v)$ in $G$. Consider the directed spanning tree polynomial
\[
\ST_n(\widetilde{G}) = \sum_{\nu \in T_n} x_{2,\nu(2)}x_{3,\nu(3)}\cdots x_{n,\nu(n)},
\]
where $T_n=\{\nu:\{2,3,\ldots,n\} \mapsto \{1,2,\ldots,n\} ~|~ \forall i ~\exists k \ \nu^{k}(i)=1 \ ; \forall i  \ (i,\nu(i)) \in E(\widetilde{G}) \}$. We note that the maps in $T_n$ correspond to directed spanning trees rooted at $1$ and every monomial $\kappa$ of $\ST_n$ is of the form $ x_{2,i_2}x_{3,i_3}\cdots x_{n,i_n}$. It is well-known that for every $G$, $\ST_n(\widetilde{G})$ can be computed even by an algebraic branching program of size $\poly(n)$ \cite{Moon70} via a determinant computation \cite{MV97}. Jerrum and Snir showed that if $G$ is the complete graph, then any monotone circuit for $\ST_n(\widetilde{G})$ must be of size $2^{\Omega(n)}$ \cite{JS82}. Note that, in this case the number of variables is $n^2$.
In contrast, we show the following.
\begin{restatable}{thm}{strongly-exponential}
\label{thm:spanning-tree}
For a sufficiently large constant $d$, let $G$ be a $d$ regular expander graph on $n$ vertices with $\lambda_2 \leq d^{1-\epsilon}$ for some $\epsilon > 0$. Then every monotone circuit for $\ST_n(\widetilde{G})$ must be of size at least $2^{\Omega(n)}$.
\end{restatable}

Now we discuss our second main result. 
Hrube\v{s} showed that if a polynomial $f_n$ of degree $d$ over $n$ variables $\{x_1, \ldots,x_n\}$ is computed efficiently by a circuit, with access to negations, of size $s$, then there exists $\epsilon_0 > 0$, such that for every $\epsilon \le \epsilon_0$, the function $F_n + \epsilon\cdot f_n$ has efficient monotone circuits, where $F_n \coloneqq (1+x_1+\cdots+x_n)^d$ is the polynomial that contains all monomials of degree at most $d$. So, proving monotone circuit lower bound for $F_n + \epsilon\cdot f_n$ for sufficiently small $\epsilon$ would lead to general circuit lower bound for $f_n$.


Most monotone lower bound techniques in the literature are crafted on arguments that cannot distinguish between two polynomials as soon as they are supported on the same set of monomials. Such arguments, not factoring in the set of coefficient values of a target polynomial, cannot obviously work here. An attractive feature of Hrube{\v s}' approach is that the choice of the polynomial $F_n$ is rather flexible. If the target polynomial is set-multilinear over the variables $\{x_{1,1},\ldots, x_{1,m}, \ldots, x_{n,1}, \ldots, x_{n,m}\}$, the choice of $F_n$ could be replaced by $$F_{n,m} \coloneqq \prod_{i=1}^n (x_{i,1}+\cdots+x_{i,m}).$$ 
Indeed, Hrube{\v s} suggests to prove lower bounds using $F_{n,m}$ as the easy polynomial, and setting $f_n$ as the Permanent polynomial.
Such monotone $\epsilon$-sensitive lower bounds would then yield commensurate lower bounds on the size of set-multilinear circuits computing $f_n$. Proving super-polynomial lower bounds even on the size of set-multilinear circuits has remained a longstanding open problem. 

In a recent work, the first $\epsilon$-sensitive lower bound was shown for a family of explicit polynomials in $\VNP$ \cite{CDM21} using discrepancy based arguments. Roughly speaking, the range for the parameter $\epsilon$ that they can handle is $\epsilon\geq 2^{-\Omega(m)}$ where the number of variables in the target set-multilinear polynomial is $O(nm)$ and 
$m\log m \leq n/4$. So the parameter $\epsilon$ is (roughly) at least $2^{-\Omega(\sqrt{\widetilde{n}})}$ when $\tilde{n}$ is the number of input variables. It was not clear at all if $\epsilon$-sensitive lower bounds could be proved when the target polynomial $f_n$ is in VP. Our second main result gives an affirmative answer by establishing exponential size lower bounds for the family of spanning tree polynomials which are  in $\VP$.  



\begin{restatable}{thm}{sensitivespanningtree}
\label{thm:spanning-tree-sensitive}
There exists a constant $\eta > 0$ such that both the polynomials $F_{n-1,n} - \epsilon\cdot \ST_n$ and $F_{n-1,n} + \epsilon\cdot \ST_n$ have monotone circuit complexity $2^{\Omega(n)}$ provided $\epsilon \ge 2^{-\eta n}$, where $\ST_n$ is the spanning tree polynomial defined over the complete graph of $n$ vertices.
\end{restatable}

As far as we can tell, there are just two instances in the past where discrepancy has been used for proving monotone lower bounds in arithmetic complexity. First, by Raz and Yehudayoff \cite{RY11} to establish strongly exponential lower bounds, among other things. Second, very recently, by Chattopadhyay et.al. \cite{CDM21} who used it to prove $\epsilon$-sensitive monotone lower bounds. In both instances, arguments seemingly specific to the target polynomial was used. A conceptual contribution of our argument is to formulate a simple but quite general transfer principle, that translates discrepancy bounds under a \emph {universal distribution} to the setting of $\epsilon$-sensitive lower bounds for set-multilinear polynomials. To describe this, let $X$ denote the matrix of variables $\{x_{i,j}\}_{1\leq i\leq n, \leq j\leq m}$. A monomial $\kappa_{\nu}$ over $X$ is naturally identified with a map $\nu : [n]\rightarrow [m]$ such that $\kappa_{\nu} = \prod_{i=1}^n x_{i,\nu(i)}$. Let $P$ be a nearly-balanced partition of $[n]$ into the sets $A$ and $B$ (i.e. $n/3 \leq |A|, |B|\leq 2n/3$). Alice and Bob respectively get two maps $\tau : A\rightarrow [m]$ and $\theta : B\rightarrow [m]$. They jointly want to solve the communication problem $C_P^{f}$ which outputs $1$ if and only if $\kappa_{\tau}\cdot \kappa_{\theta}$ is a monomial in the set-multilinear polynomial $f$. Let $\Delta$ be a distribution on the space $[m]^n$ of monomials. The \emph{best partition discrepancy} w.r.t. $\Delta$ is defined as follows:  
\[\disp_{\Delta}(C^f) \coloneqq \max\limits_{\substack{(A,B)\ \text{is a nearly-balanced}\\{\text{partition of}}\ [n]}}\{\disp_{\Delta_A,\Delta_B}(C_P^f)\}.\]

Here $\Delta_A$ and $\Delta_B$ are the induced distributions by $\Delta$ on $[m]^A$ and $[m]^B$ respectively.  
Small discrepancy with respect to every balanced partition translates to the following lower bound statement. 
\begin{restatable}[Discrepancy-Sensitivity Correspondence]{thm}{Discrepancy}
\label{thm-sensitivity-intro}
Consider a distribution $\Delta$ over $[m]^n$. Then, the monotone complexity of $F_{n,m} - \epsilon\cdot f$  ($\text{resp}. ~F_{n,m} + \epsilon \cdot f$) is at least $\frac{\epsilon}{3\gamma}$ ($\text{resp}.\ \frac{\epsilon}{6\gamma}$) as long as $\epsilon \geq\frac{6\gamma}{1-3\gamma}$ ($\text{resp.} ~\epsilon\geq \frac{6\gamma}{1-12\gamma}$), where $\gamma\coloneqq \disp_{\Delta}(C^f)$.
\end{restatable}

The reason we need to handle the discrepancy for every nearly-balanced partition stems from the fact that in the decomposition theorem of monotone circuits each product polynomial $a\cdot b$ appears with its own set of nearly-balanced partition of the variables. Now to use the theorem above, all we need to do is to craft a distribution $\Delta$ for the spanning tree problem such that the discrepancy remains small (under $\Delta$) for every nearly-balanced partition $P$.   

We give a sketch of the key ideas. Firstly, consider the following communication problem in the basic 2-party model (spanning tree problem): Given a vertex set $V$ partitioned into two nearly-balanced sets $V_1$ and $V_2$ ($|V_1|, |V_2|\geq |V|/3$), Alice and Bob receive 
two functions $\tau : V_1\rightarrow V$ and $\theta: V_2\rightarrow V$. They want to jointly decide whether the oriented edges given by $\tau$ and $\theta$ form a spanning tree rooted at a special vertex designated as the root. Using a graph gadget, we show that the standard inner product function can be reduced to this problem via a rectangular reduction. 

If the inputs for the inner product function are sampled according to the uniform distribution $U_{n}$, the reduction induces a distribution $\mathcal{D}_P$ on the set of functions $(\tau,\theta)$ when the partition $P$ is fixed.
Using the low discrepancy of inner product under the uniform distribution, one can directly infer that the spanning tree problem also has low discrepancy under the distribution $\mathcal{D}_P$. 

However to make it useful in the context of monotone circuit lower bounds, we need to find a \emph{single distribution} that handles \emph{all} possible (nearly-balanced) partitions as mandated by our Discrepancy-Sensitivity Correspondence. For instance, if Inner-Product is defined as $x_1 y_1 + \cdots + x_n y_n \text{ (mod }2\text{)}$, then although the partition in which Alice gets $x_1,\ldots,x_n$ and Bob $y_1,\ldots,y_n$ has exponentially small discrepancy w.r.t the uniform distribution, the perfectly balanced partition in which Alice gets $x_1,y_1,\ldots,x_{n/2},y_{n/2}$ and Bob the remaining $n/2$ pairs has large, i.e. $O(1)$ discrepancy w.r.t. \emph{every} distribution. 

Nevertheless, we are able to show exponentially small discrepancy for the spanning tree problem w.r.t. every nearly-balanced partition, by a random embedding of our gadget graph in the complete graph on $n$ vertices. This enables designing a suitable distribution $\Delta$ on the set of maps $[n]^{n-1}$ as stated in the following lemma that is the main technical contribution of our second result.

\begin{lemma}[Informal]\label{lem:variable-dist-intro}
There exists a  distribution $\Delta$ on the set of maps from $\{2,3,\ldots,n\}$ to $[n]$ such that $\disp_{\Delta}\big(C^{\ST_{n}}\big)$ is $2^{-\Omega(n)}$.
\end{lemma}
This result could be of independent interest even beyond algebraic complexity theory. 


\subsection*{Organization.}
In Section \ref{sec:prelim}, we collect some basic results from monotone arithmetic computation and from communication complexity. We prove a strongly exponential monotone circuit size lower bound for the spanning tree polynomial in Section \ref{sec:spanning-tree}. 
In Section \ref{sec:sensitive}, we develop a general framework for proving $\epsilon$-sensitive lower bounds for set-multilinear polynomials via a suitable discrepancy based technique. We prove the low discrepancy of spanning tree problem under a certain universal distribution in Section \ref{sec: low-discrepancy}. 
We then use the low discrepancy result to complete the $\epsilon$-sensitive lower bound result for the spanning tree polynomial in Section \ref{sec:lb-spanning-tree-polynomial}. We state some open problems in Section \ref{sec:openproblem}.

\section{Preliminaries}\label{sec:prelim}
\subsubsection*{Notation} 
Let $[n]=\{1,2,\ldots, n\}$. Polynomials are always considered over $\R[X]$ where $\R$ is the set of reals. 
For a polynomial $p$, let $\var(p)$ denote the set of variables in $p$. 

\subsubsection*{Set-multilinear Polynomials}
Let $X=\cup_{i=1}^n X_i$ be a set of variables where $X_i=\{x_{i,1}, x_{i,2}, \ldots, x_{i,m}\}$. A polynomial $p\in\R[X]$ is set-multilinear if each monomial in $p$ respects the partition given by the set of variables $X_1, X_2, \ldots, X_n$. In other words, each monomial $\kappa$ in $p$ is of the form $x_{1, j_1} x_{2, j_2}\cdots x_{n, j_n}$. 


\subsubsection*{Ordered Polynomial} 
For a monomial of the form $\kappa=x_{i_1, j_1} x_{i_2, j_2} \cdots x_{i_n, j_n}$ we define the set $I(\kappa)=\{ i_1,i_2,\ldots,i_n\}$. If a polynomial $p$ has the same set $I(\kappa)$ for every monomial occurring it it with a nonzero coefficient, then we say that the polynomial is ordered and we write $I(p)=I(\kappa)$ for each $\kappa$. Clearly, the set-multilinear polynomials are ordered polynomials with $I(p)=\{1,2,\ldots,n\}$. 

\subsubsection*{Structure of Monotone Circuits}
The main structural result for monotone circuits that we use throughout, is the following theorem. 

\begin{theorem}\cite[Lemma 1]{Yehudayoff19}\label{thm:monotone-structure}
Let $n>2$ and $p \in \mathbb{R}[X]$ be an ordered monotone polynomial with $I(p)=[n]$. Let $C$ be a monotone circuit of size $s$ that computes $p$.  Then, we can write 
$$p=\sum^{s}_{t=1} a_t \cdot b_t $$
where $a_t$ and $b_t$ are monotone ordered polynomials with $\frac{n}{3} \leq |I(a_t)| \leq \frac{2n}{3}$ and $I(b_t)= [n]\setminus I(a_t)$. Moreover, $a_t b_t\leq p$ for each $1\leq t\leq s$, by which we mean that the coefficient of any monomial in $a_t b_t$ is bounded by the coefficient of the same monomial in $p$.
\end{theorem}
A partition $P=(A,B)$ of $[n]$ is said to be perfectly balanced if $|A|=|B|=\frac{n}{2}$ and is said to be nearly-balanced if $\frac{n}{3}\leq|A|,|B|\leq\frac{2n}{3}$. 
An ordered product polynomial $a\cdot b$ on $n$ variables is said to be nearly-balanced if $\frac{n}{3}\leq|I(a)|,|I(b)|\leq \frac{2n}{3}$.

\subsection*{Communication Complexity}
We recall some basic results from communication complexity. The details can be found in \cite{KN06}.
Let us very briefly first recall the basic notions in the 2-party communication model of Yao. The joint input space of Alice and Bob is $\{0,1\}^m \times \{0,1\}^m$ with each player receiving an $m$-bit Boolean string, and they want to evaluate a Boolean function $F : \{0,1\}^m \times \{0,1\}^m \to \{0,1\}$. One defines a combinatorial rectangle $R$ as a product set $A \times B$, for some $A,B \subseteq \{0,1\}^m$. Put another way, $R$ is just a submatrix of the $2^m \times 2^m$ communication matrix $M_F$ of the function $F$, that Alice and Bob want to compute. The rows of this matrix are indexed by possible inputs of Alice and the columns by the ones of Bob and $M_F(x,y) = F(x,y)$. One of the important notions is discrepancy. For a rectangle $R$, the discrepancy $\disp_{\delta}(F,R) \coloneqq \left| \delta(R\cap F^{-1}(0)) - \delta(R\cap F^{-1}(1))\right|$ where $\delta$ is a distribution on the input space $\{0,1\}^m\times \{0,1\}^m$. The discrepancy of $F$ under $\delta$ is defined as 
$$\disp_{\delta}(F)\coloneqq\max_{R}\disp_{\delta}(F,R)$$.

The inner product function $\IP_m(x,y)\coloneqq\sum_{i=1}^m x_i y_i \mod 2$. It is well-known that the discrepancy of the inner product function is small under the uniform distribution $U_{m}$ over $\{0,1\}^m\times \{0,1\}^m$. This was first proved by Chor and Goldreich \cite{CG88}. A self-contained proof can be found in \cite{KN06}.  

\begin{theorem}\cite[Example  3.29]{KN06}\label{thm:ip-discrepancy}
Under the uniform distribution $U_{m}$ over $\{0,1\}^m\times \{0,1\}^m$, $\disp_{U_{m}}(\IP_m)=2^{-\Omega(m)}$.
\end{theorem}

\section{Strong Exponential Separation of VP and  Monotone VP}  \label{sec:spanning-tree}
In this section we prove Theorem \ref{thm:spanning-tree}. 
For a graph $G$, let $V(G), E(G)$ denote the set of vertices and edges of $G$ respectively, and for any pair $S,T \subseteq V(G)$, let $E(S,T) \equiv \{(u,v) \in E(G)\,:\, u \in S\,, v\in T\}$.
\begin{lemma}[Expander Mixing Lemma]\label{lem:expmixing}\cite[Lemma 2.5]{HLW06}
Let $G$ be an undirected $d$ regular graph such that $\lambda_2$ is the second largest eigenvalue of the adjacency matrix of $G$. Then, for every $S,T \subseteq V(G)$
$$\bigg||E(S,T)| - \frac{d}{n} |S||T| \bigg| \leq \lambda_2 \sqrt{|S| |T|}.$$
\end{lemma}
We also need Matrix Tree Theorem which we state below. 
\begin{theorem}\label{thm:matrix-tree}[Matrix Tree Theorem]\label{matrixtree}\cite[Theorem 13.1]{MM11}
Let $G$ be an undirected graph on $n$ vertices and let $0,\mu_1,\mu_2,\ldots,\mu_{n-1}$ be the eigenvalues of the Laplacian of $G$. Then the number of spanning trees in $G$ is $\frac{1}{n}\mu_1 \cdot \mu_2\cdots \mu_{n-1}$.
\end{theorem}

\begin{proof}[Proof of Theorem~\ref{thm:spanning-tree}]

Consider a family of $d$-regular expander graphs where $d$ is a sufficiently large constant and the second largest eigenvalue is bounded by $d^{1-\epsilon}$ for a suitable $\epsilon >0$. For example, the current proof works for $\epsilon=0.25$ and such a family of graphs can be explicitly constructed \cite{RVW02}. 
Let $G=G_n$ be the $n^{th}$ graph in the family. 

Suppose $\ST_n(\widetilde{G})$ has a monotone circuit of size $S$. Recall from Section \ref{sec:intro} that $\widetilde{G}$ is the graph obtained from $G$ which has edges in both directions for each undirected edge in $G$. Then applying Theorem~\ref{thm:monotone-structure} to the polynomial $\ST_n(\widetilde{G})$ we get
\begin{equation}\label{strapply}
\ST_n(\widetilde{G})=\sum^{S}_{s=1} a_s b_s.
\end{equation}

For a fixed $s$, let $X_t = \{x_{t,j}| x_{t,j} \in \var(a_s) \cup \var(b_s) \}$. Since every monomial of $ST_n(\widetilde{G})$ has distinct first indices we conclude that $I(a_s) \cap I(b_s) = \emptyset$.

Now we upper bound $\sum^{n}_{t=2}|X_t|$. We note that if $i \in I(a_s)$ and $j \in I(b_s)$ then it cannot be the case that both $x_{i,j}$ and $x_{j,i}$ are in $\cup^{n}_{t=2} X_t$. Suppose $x_{i,j} , x_{j,i} \in \cup^{n}_{t=2} X_t$ then it must be the case that $x_{i,j} \in \var(a_s)$ and $x_{j,i} \in \var(b_s)$ (since $i \not \in I(b_s)$ and $j \not \in I(a_s)$). Then some monomial in $a_s b_s$ contains $x_{i,j} x_{j,i}$ which is a two cycle and cannot be part of the spanning tree polynomial. 

This shows that in the set of undirected edges $E(I(a_s),I(b_s))$, at least one out of the two directed edge variables, corresponding to an undirected edge, must be absent in $\cup^n_{t=2} X_t$. Thus we may bound,

$$\sum^{n}_{t=2}|X_t| \leq dn - \big|E(I(a_s),I(b_s))\big|.$$

Since $G$ is an expander, using Lemma~\ref{lem:expmixing} we conclude that
$$\bigg|\big|E(I(a_s),I(b_s))\big| - \frac{d}{n} |I(a_s)| |I(b_s)| \bigg| \leq \lambda_2 \sqrt{|I(a_s)| |I(b_s)|}.$$
On rearranging, we obtain
$$\bigg|E(I(a_s),I(b_s))\bigg| \geq \frac{d}{n} |I(a_s)||I(b_s)| - \lambda_2 \sqrt{|I(a_s)| |I(b_s)|}.$$
Since $|I(a_s)|,|I(b_s)| \geq \frac{n}{3}$ and $|I(a_s)|+|I(b_s)|=n$ we may simplify the right hand side as $$\bigg|E(I(a_s),I(b_s))\bigg| \geq \frac{d}{n} \frac{n^2}{9} - \lambda_2 \frac{n}{2} = n(\frac{d}{9} - \frac{\lambda_2}{2}).$$

Since $\lambda_2 \leq d^{1 - \epsilon}$, we may relax the right hand side and write $|E(I(a_s),I(b_s))| \geq \frac{nd}{18}$ for sufficiently large $d$. 
Let $\alpha=\frac{1}{18}$. 
Now we bound the total numbers of monomials in $a_s b_s$ as

$$|\mon(a_s b_s)| \leq \prod^n_{t=2} |X_t| \leq \left(\frac{\sum^n_{t=2} |X_t|}{n-1}\right)^{n-1} \leq ((1-\alpha) \frac{nd}{n-1})^{n-1} \leq (1.01 d (1-\alpha))^{n-1}$$
for sufficiently large $n$.  

Then, the number of monomials in $\ST_n(\tilde{G})$:

\begin{equation}\label{upper-bd}
|\mon(\ST_n(\widetilde{G}))| \leq S (1.01 d (1-\alpha))^{n-1}.
\end{equation}

Let $L(G)$ be the Laplacian of the graph $G$ with eigenvalues $0 < \mu_1 \leq \mu_2 \leq \ldots \leq \mu_{n-1}$. Since $G$ is an expander, we conclude that $\mu_1 \geq (d-\lambda_2)$. Then, Theorem~\ref{matrixtree} implies that 
$$|\mon(\ST_n(\widetilde{G}))| = \frac{1}{n}\mu_1 \mu_2 \cdots \mu_{n-1} \geq \frac{1}{n}(d - \lambda_2)^{n-1} \geq \frac{1}{n}(d - d^{1-\epsilon})^{n-1}.$$

\begin{remark}
Notice that each spanning tree rooted at the vertex $1$ in $G$ is in bijective correspondence with a rooted tree at the vertex $1$ in $\widetilde{G}$. 
\end{remark}

Putting the above bound together with the upper bound in Equation \ref{upper-bd}, we get that
$$\frac{1}{n}(d - d^{1-\epsilon})^{n-1} \leq |\mon(\ST_n(\widetilde{G}))| \leq S (1.01 d (1-\alpha))^{n-1}. $$
This immediately implies that $S\geq \frac{1}{n}\left(\frac{d - d^{1-\epsilon}}{1.01 d (1-\alpha)}\right)^{n-1} \geq \frac{1}{n}(\frac{99}{101 (1-\alpha)})^{n-1} = 2^{\Omega(n)}$, for sufficiently large $d$.
\end{proof}

\section{Discrepancy Implies $\epsilon$-Sensitive Bounds}
\label{sec:sensitive}

In this section we formulate a simple method to prove ${\epsilon}$-sensitive lower bound against set-multilinear polynomials by showing a connection with discrepancy under a universal distribution.\\\\
Consider the input matrix $X$ of dimension $n \times m$ with entries $X[i,j]\coloneqq  x_{i,j}$ of indeterminates. Define $\mathbb{M}[X]$ to be the set of all set-multilinear monomials of degree $n$ over variable set $X=\{X_i|\  i \in [n]\}$ and $\forall i$ $X_i=\{x_{i,j}|\ j\in [m]\}$.
We identify a monomial $\kappa_{\nu}\in \mathbb{M}[X]$ with a map $\nu: [n]\to [m]$ in the following way, $\kappa_{\nu}= \prod_{i=1}^n x_{i,\nu(i)}$. This forms a bijection between set $\mathbb{M}[X]$ and $\Upsilon=\{\nu |\ \nu:[n]\to [m]\}$.\\ Given a partition $P=(A,B)$ of $[n]$ which is nearly-balanced i.e. $|A|,|B| \in [\frac{n}{3},\frac{2n}{3}]$, 
define sets of maps $\mathcal{A}_P=\{\tau |\ \tau: A \to [m]\}$ and $\mathcal{B}_P=\{\theta|\ \theta: B\to [m]\}$. Clearly any map $\tau \in \mathcal{A}_P$ ($\theta\in \mathcal{B}_P$) gives a set-multilinear monomial $\kappa_{\tau}$ ($\kappa_{\theta}$) of degree $|A|$ ($|B|$) over variable set $X_A=\{X_i| i\in A\}$ ($X_B=\{x_j| j\in B\}$) in the following way, $\kappa_{\tau}=\prod_{i\in A}x_{i,\tau(i)}$ ($\kappa_{\theta}=\prod_{j\in B} x_{j,\theta(j)}$).\\\\
For any degree $n$ set-multilinear polynomial $f$ with $0-1$ coefficients over variable set $X$ and a nearly-balanced partition $P=(A,B)$ of $[n]$, define the communication function $C_P^f: \mathcal{A}_P\times \mathcal{B}_P \to \{0,1\}$ as follows:

$$C_P^f(\tau,\theta) \coloneqq
\begin{cases}
    1&\text{iff coefficient of}\ \kappa_{\tau}\cdot\kappa_{\theta}\ \text{is}\ 1\ \text{in}\ f\\
    0&\text{iff coefficient of}\ \kappa_{\tau}\cdot\kappa_{\theta}\ \text{is}\ 0\ \text{in}\ f 
\end{cases}
$$
Naturally, $C_P^f$ can be viewed as a Boolean communication matrix with rows indexed by elements in $\mathcal{A}_P$ and columns indexed by elements in $\mathcal{B}_P$. Let us consider a distribution $\Delta$ over $[m]^n$. Given a partition $P=(A,B)$ of $[n]$, the distribution $\Delta$ induces two distributions over the rows and columns of $C_P^f$. We denote the induced distributions by $\Delta_A$, $\Delta_B$ respectively. Then, let \[\disp_{\Delta}(C^f) \coloneqq \max\limits_{\substack{(A,B)\ \text{is a nearly-balanced}\\{\text{partition of}}\ [n]}}\{\disp_{\Delta_A,\Delta_B}(C_P^f)\}.\]\\
Consider the full set-multilinear polynomial $F_{n,m}= \prod_{i=1}^n\sum_{j=1}^m x_{i,j}$. Now define the polynomials $g\coloneqq F_{n,m}-\epsilon\cdot f$ and $g'\coloneqq F_{n,m}+\epsilon\cdot f$.  Observe that in the polynomial $g$ ($g'$) all those monomials which have coefficients $1$ in $f$  have coefficients exactly ${1-\epsilon}$ ($1+\epsilon$) and those which have coefficients $0$ in $f$ have  coefficients exactly $1$. 
The main result of this section is the proof of Theorem \ref{thm-sensitivity-intro}. For the reader's convenience, we restate the theorem below.

\Discrepancy*




\begin{proof}
We present the proof for $g$. The proof for $g'$ is analogous. The main idea is to define a measure $\mathcal{M}$ on the set of monomials. Then we extend the measure to the polynomial by linearity. We define the measure $\mathcal{M}$ in the following way.

   $$ \mathcal{M}(\kappa)\coloneqq
\begin{cases}
    \Delta(\kappa)&\text{if coefficients of}\ \kappa\ \text{is}\ 1\ \text{in}\ f\ \text{and}\ \kappa\  \text{is in}\ \supp(\Delta)\\
    -\Delta(\kappa)&\text{if coefficients of}\ \kappa\ \text{is}\ 0\ \text{in}\ f\ \text{and}\ \kappa\  \text{is in}\ \supp(\Delta)\\
    0          &\kappa\ \text{is not in}\ \supp(\Delta)       
\end{cases}    
$$ 
Now we need the following two lemmas to prove Theorem \ref{thm-sensitivity-intro}.
\begin{lemma}\label{upper-bound-lemma}
For every distribution $\Delta$ over $[m]^n$ and for any nearly-balanced ordered product polynomial $a\cdot b$ of degree $n$ with $||a\cdot b||_{\infty}\leq 1$, the measure 
$|\mathcal{M}(a\cdot b)|\leq \disp_{\Delta_A,\Delta_B}(C_P^f)$ where $A=I(a),B=I(b)$.
\end{lemma}

\begin{lemma}\label{lower-bound-lemma}
let $f$ be a set-multilinear polynomial of degree $n$ and $g=F_{n,m}-\epsilon\cdot f$ with $\epsilon \geq \frac{6\gamma}{1-3\gamma}$ where $\gamma=\disp_{\Delta}(C^f)$. Then $|\mathcal{M}(g)| \geq \frac{\epsilon}{3}$.
\end{lemma}
Let us first prove Theorem \ref{thm-sensitivity-intro} assuming Lemmas \ref{upper-bound-lemma} and \ref{lower-bound-lemma}. Further, assume $g$ has a monotone circuit of size $s$. Then, using Theorem \ref{thm:monotone-structure}, $g$ can be expressed as $g=\sum_{i=1}^s a_i\cdot b_i$ where each $a_i\cdot b_i$ is a nearly-balanced product polynomial. Using Lemmas~\ref{upper-bound-lemma} and \ref{lower-bound-lemma}, and the sub-additivity of the measure $\mathcal{M}$, it follows that,
\[
\frac{\epsilon}{3}\leq_{\substack{\text{Lemma}\  \ref{lower-bound-lemma}}} |\mathcal{M}(g)| =_{\text{Definition}} |\mathcal{M}(\sum_{i=1}^s a_i\cdot b_i)|\leq_{\substack{\text{Sub-additivity}} \sum_{i=1}^s |\mathcal{M}(a_i\cdot b_i)| \leq_{\substack{\text{Lemma}\  \ref{upper-bound-lemma}}}}  s \cdot \gamma
\]This shows $s \geq \frac{\epsilon}{3\gamma}$ with $\epsilon\geq \frac{6\gamma}{1-3\gamma}$. 

\end{proof}
Now we provide proofs of Lemmas \ref{lower-bound-lemma} and \ref{upper-bound-lemma}.
To prove Lemma \ref{upper-bound-lemma}, it'd be convenient to write a product polynomial as a sum of \emph{rectangular} product polynomials, via the following lemma.

\begin{lemma}
\label{rect-poly}
Let $a\cdot b$ be a product polynomial with both $a$ and $b$ are monotone. Then $a\cdot b$ can be written as a sum of rectangular product polynomials, i.e. $$a\cdot b=\sum\limits_{i,j} \lambda_{i,j} ~\alpha_i\cdot \beta_j.$$ Here $\sum\limits_{i,j} \lambda_{i,j}=||a\cdot b||_{\infty}$ and $\lambda_{i,j}>0$. Moreover, the coefficients of monomials in $\alpha_i$ and $\beta_j$ are only $0,1$. 
\end{lemma}

\begin{proof}
    First we want to show if we take any monotone polynomial $a$, we can write it as $\sum\limits_i \lambda_i a_i$ where $a_i$s are monotone polynomials with coefficients $0,1$ and $\sum\limits_i \lambda_i=||a||_\infty$, and the $\lambda_i$s are non-negative. For this, write polynomial $a$ as a vector $V_a$ of coefficients, i.e. the vector has a coordinate for every set-multilinear  monomial appearing in the polynomial. Let the vector $V_a=[p_1,p_2,\cdots,p_\ell]^{T}$ where each $p_i>0$.

    We prove the claim using induction on the number of nonzero entries in the vector. \\ The base case will be $a$ being a single monomial, for which the statement is trivially true.
    
    For the inductive step, consider $p=\min_i\{p_i\}$. 
    Let $V_p$ be the vector of length $\ell$ such that each entry is $p$. 
    Let $\hat{f}$ be the polynomial having the same monomial set as $a$ and the coefficients are all $1$. 
      Now consider the polynomial $a'$ corresponding to the vector $V_{a'}=V_a-V_p$, where the difference has been taken coordinate wise. This new vector has $0$ valued coordinates for which $V_a$ has $p$. Now by inductive hypothesis $a'=\sum\limits_j \lambda_j a'_j$ where $a'_j$ are polynomials having $0-1$ coefficients and $\sum\limits_j \lambda_j=||a'||_{\infty}$. Observe that  $||a'||_{\infty}=||a||_{\infty}-p$.
    
    Hence, $$a=a'+p\cdot \hat{f}=\sum\limits_j \lambda_j a'_j+p\cdot \hat{f}.$$ Here $\sum_j\lambda_j+p=||a'||_{\infty}+p=||a||_{\infty}.$ This proves our first claim.
    
    Using the above, write $a=\sum\limits_i\lambda_i\cdot \alpha_i$ and $b=\sum\limits_j\lambda'_j\cdot \beta_j$. Where $\sum\limits_i \lambda_i=||a||_\infty$ and $\sum\limits_j \lambda'_j=||b||_\infty$. So, $$a\cdot b= \sum\limits_{i,j} \lambda_i\lambda'_j  \alpha_i\cdot \beta_j, $$ and $\sum\limits_{i,j}\lambda_i\cdot\lambda'_j=(\sum\limits_i\lambda_i)(\sum\limits_j\lambda'_j)=||a||_{\infty}\cdot||b||_{\infty}=||a\cdot b||_{\infty}$.
\end{proof}

Now we are ready to prove Lemma \ref{upper-bound-lemma}.

\begin{proof}
By Lemma \ref{rect-poly}, decompose 
$a\cdot b= \sum\limits_{i,j} \lambda_{i,j} ~\alpha_i\cdot \beta_j$, where $\sum\limits_{i,j} \lambda_{i,j}=||a\cdot b||_{\infty}$ and $\lambda_{i,j}>0$. Hence by triangle inequality and sub-additivity, 
\[
|\mathcal{M}(a\cdot b)|\leq \sum\limits_{i,j}\lambda_{i,j} |\mathcal{M}(\alpha_i\cdot \beta_j)|. 
\]Now $\mathcal{M}(\alpha_i\cdot \beta_j)= \disp_{\Delta_A,\Delta_B}(R)$ where $R$ is the rectangle $A\times B$ with $A$ is indexed by the set of ordered monomials in $\alpha_i$  and $B$ is indexed by the set of ordered monomials in $\beta_j$ and partition $P=(A,B)$ of $[n]$ is given by $A=I(\alpha_{i})$ and $B=I(\beta_j)$.  So we have that $|\mathcal{M}(\alpha_i\cdot \beta_j)|\leq \disp_{\Delta_A,\Delta_B}(R)\leq \disp_{\Delta}(C^f)=\gamma$. 
So,
\[
|\mathcal{M}(a.b)| \leq \sum\limits_{i,j}\lambda_{i,j} |\mathcal{M}(\alpha_i\cdot \beta_j)|\leq \sum\limits_{i,j}\lambda_{i,j} \gamma\leq \gamma
\]
The third inequality comes from Lemma \ref{rect-poly}, i.e. $\sum\limits_{i,j}\lambda_{i,j}=||a\cdot b||_{\infty}\leq 1$. So the proof follows.    
\end{proof}
Next we prove the Lemma \ref{lower-bound-lemma}.

\begin{proof}
Since $g= F_{n,m}-\epsilon\cdot f$, applying the measure $\mathcal{M}$ we have,
$$\mathcal{M}(g)=\mathcal{M}(F_{n,m}) - \epsilon \cdot \mathcal{M}(f).$$
Since the full polynomial $F_{n,m}$ is a product of linear forms, by using Lemma~\ref{upper-bound-lemma} we conclude that $|\mathcal{M}(F_{n,m})|\leq \gamma$.

Let $S_1, S_2$ be the sets of monomials with coefficients $1$ and $0$ in $f$ that are in the support of $\Delta$, then we observe
\[\bigg|\mathcal{M}(F_{n,m})\bigg| = \bigg|\sum_{\kappa\in S_1}\Delta(\kappa) - \sum_{\kappa\in S_2} \Delta(\kappa)\bigg| \leq \gamma
\]
Further,  
\[\sum_{\kappa\in S_1}\Delta(\kappa) + \sum_{\kappa\in S_2} \Delta(\kappa)=1\]
Combining the above two equations we conclude that $$\mathcal{M}(f)= \sum_{\kappa\in S_1} \Delta(\kappa) \geq \frac{1}{2}-\frac{\gamma}{2}.$$
Using this estimate, $|\mathcal{M}(g)| \geq \epsilon \cdot (\frac{1}{2}-\frac{\gamma}{2}) - \gamma$. This will be at least $\geq \frac{\epsilon}{3}$ for $\epsilon \geq\frac{6\gamma}{1-3\gamma}$
\end{proof}

\section{Low Discrepancy of Spanning Tree Problem}\label{sec: low-discrepancy}

In this section, we consider a communication problem that decides whether two sets of oriented edges of a given vertex set distributed among Alice and Bob form a spanning tree. 
By embedding the inner product function in it, we show that the problem has low discrepancy. 

\begin{problem}{Spanning tree problem ($\comst_n$).} 

The input to this problem is a fixed vertex set $V \cup \{r\}$ where $V$ is partitioned into $V_1, V_2$ and $|V|=n$. Alice gets $V_1$ and Bob gets $V_2$ respectively. Additionally, Alice (Bob) is given a map $\tau: V_1\rightarrow V\cup \{r\}$ ($\theta : V_2 \rightarrow V\cup  \{r\}$) specifying the directed edges $i\rightarrow \tau(i)$ ($j\rightarrow \theta(j)$) for $i \in V_1$ ($j \in V_2$).
The communication problem is to decide whether the edges $\{(i,\tau(i)) : i\in V_1\}\cup \{(j,\theta(j)) : j\in V_2\}$ form a spanning tree rooted at $r$ or not.  

\end{problem}

Corresponding to the partition $P=(V_1,V_2)$, we define a boolean function $g_P : \Tau\times \Theta \rightarrow \{0,1\}$ where $\Tau=\{\tau : V_1\rightarrow V\cup \{r\}\}$ and $\Theta = \{\theta : V_2 \rightarrow V\cup \{r\}\}$. We set $g_P(\tau,\theta) = 1$ if and only if $\{(i,\tau(i)) : i\in V_1\}\cup \{(j,\theta(j)) : j\in V_2\}$ forms a spanning tree rooted at $r$. 
In this work, unless explicitly stated otherwise, the partition ($P$) will always be nearly-balanced, i.e.\ $|V_1|,|V_2|\geq \frac{|V|}{3}$. 
\begin{remark}
In the definition of Problem 4.1, the functions $\tau$ and $\theta$ ensure that every vertex in the resulting directed graph has out-degree $1$ except the vertex $r$. So, if the oriented edges form a directed spanning tree then it must be rooted at $r$.
\end{remark}
\begin{remark}
The communication matrix of $g_P$, $M_P$ has rows indexed by elements in $\Tau$ and columns are indexed by elements in $\Theta$. For any entry $(\tau,\theta)$ in the matrix,  $M_P[\tau,\theta]=1$ if and only if $g_P(\tau,\theta)=1$. 
\end{remark}
In the next section we are going to prove the following lemma.
\begin{lemma}\label{lem:fixed-dist}
For each perfectly balanced partition $P$ there exists a distribution $\mathcal{D}_P$ on $\Tau \times \Theta$ such that the discrepancy of $g_P$ under the distribution $\mathcal{D}_P$ is $2^{-\Omega(n)}$.
\end{lemma}
The argument to establish the low discrepancy of $g_P$ under a certain distribution uses a rectangular reduction from the inner product problem to the spanning tree problem which we describe next.  

\subsection{The Reduction}\label{gadget reduction}
Consider a vertex set $V\cup \{r\}$ with $|V| = 4n+2$, where $n\geq 1$ is an integer. The set $V$ is partitioned into perfectly balanced sets $V_1,V_2$ such that $|V_1|= 2n+1,|V_2|= 2n+1$. Alice (Bob) gets the vertex set $V_1$ ($V_2$) and a map $\tau \in \Tau$ ($\theta \in \Theta$), where  $\Tau=\{\tau : V_1\rightarrow V\cup \{r\}\}$ ( $\Theta = \{\theta : V_2 \rightarrow V\cup \{r\}\}$).  In this section we exhibit a reduction $\red$ from $\IP_n$ to $\comst_{4n+2}$ which has the following property. Given an input instance $(x,y)\in\{0,1\}^n\times\{0,1\}^n$ of $\IP_n$ the reduction $\red$ computes  $(\tau_x,\theta_y) \in \Tau \times \Theta$ such that $\IP_n(x,y)= 1 \mod 2 ${ if and only if}$\ g_P(\tau_x,\theta_y)=1$.\\
Let $G_{x,y}$ be the directed graph with vertex set $V\cup \{r\}$ and the edges are given by $\{(i,\tau_x(i)) : i\in V_1\}\cup \{(j,\theta_y(j)) : j\in V_2\}$. The graph $G_{x,y}$ has $n+1$ sub-gadgets $G_1,\ldots, G_{n+1}$. Each sub-gadget $G_i$ ($i\leq n$) has two vertices from $V_1$ and two vertices from $V_2$. For convenience, the vertices from $V_1$ are labelled as $a_{i,1},a_{i,2}$ and the vertices from $V_2$ are labelled as $b_{i,1},b_{i,2}$. The sub-gadget $G_{n+1}$ has three vertices,  $a_{n+1,1},\ \in V_1,\ b_{n+1,1}\ \in V_2 $ and $r$.  In the following figure we depict the edge orientations of the sub-gadgets.

\tikzset{every picture/.style={line width=0.75pt}} 

\begin{tikzpicture}[x=0.75pt,y=0.75pt,yscale=-.80,xscale=1]\label{fig:figure1}

\draw    (189.98,159.17) -- (305.48,158.77) ;
\draw [shift={(307.48,158.77)}, rotate = 539.8] [color={rgb, 255:red, 0; green, 0; blue, 0 }  ][line width=0.75]    (10.93,-3.29) .. controls (6.95,-1.4) and (3.31,-0.3) .. (0,0) .. controls (3.31,0.3) and (6.95,1.4) .. (10.93,3.29)   ;
\draw    (166.22,174.74) -- (216.29,242.97) ;
\draw [shift={(217.48,244.58)}, rotate = 233.72] [color={rgb, 255:red, 0; green, 0; blue, 0 }  ][line width=0.75]    (10.93,-3.29) .. controls (6.95,-1.4) and (3.31,-0.3) .. (0,0) .. controls (3.31,0.3) and (6.95,1.4) .. (10.93,3.29)   ;
\draw    (265,244.58) -- (329.88,174.62) ;
\draw [shift={(331.24,173.16)}, rotate = 492.84] [color={rgb, 255:red, 0; green, 0; blue, 0 }  ][line width=0.75]    (10.93,-3.29) .. controls (6.95,-1.4) and (3.31,-0.3) .. (0,0) .. controls (3.31,0.3) and (6.95,1.4) .. (10.93,3.29)   ;
\draw    (202.06,404.63) -- (314.41,405.41) ;
\draw [shift={(316.41,405.43)}, rotate = 180.4] [color={rgb, 255:red, 0; green, 0; blue, 0 }  ][line width=0.75]    (10.93,-3.29) .. controls (6.95,-1.4) and (3.31,-0.3) .. (0,0) .. controls (3.31,0.3) and (6.95,1.4) .. (10.93,3.29)   ;
\draw    (268.13,306.71) -- (338.86,388.34) ;
\draw [shift={(340.17,389.85)}, rotate = 229.09] [color={rgb, 255:red, 0; green, 0; blue, 0 }  ][line width=0.75]    (10.93,-3.29) .. controls (6.95,-1.4) and (3.31,-0.3) .. (0,0) .. controls (3.31,0.3) and (6.95,1.4) .. (10.93,3.29)   ;
\draw    (172.3,389.05) -- (219.59,308.43) ;
\draw [shift={(220.61,306.71)}, rotate = 480.4] [color={rgb, 255:red, 0; green, 0; blue, 0 }  ][line width=0.75]    (10.93,-3.29) .. controls (6.95,-1.4) and (3.31,-0.3) .. (0,0) .. controls (3.31,0.3) and (6.95,1.4) .. (10.93,3.29)   ;
\draw    (265,244.58) .. controls (344.08,264.78) and (391.85,371.1) .. (365.2,403.98) ;
\draw [shift={(363.93,405.43)}, rotate = 313.39] [color={rgb, 255:red, 0; green, 0; blue, 0 }  ][line width=0.75]    (10.93,-3.29) .. controls (6.95,-1.4) and (3.31,-0.3) .. (0,0) .. controls (3.31,0.3) and (6.95,1.4) .. (10.93,3.29)   ;
\draw    (268.13,306.71) .. controls (305.17,298.33) and (419.98,213.53) .. (355.98,158.41) ;
\draw [shift={(355,157.58)}, rotate = 399.78999999999996] [color={rgb, 255:red, 0; green, 0; blue, 0 }  ][line width=0.75]    (10.93,-3.29) .. controls (6.95,-1.4) and (3.31,-0.3) .. (0,0) .. controls (3.31,0.3) and (6.95,1.4) .. (10.93,3.29)   ;
\draw    (455.94,153.86) -- (563.27,153.47) ;
\draw [shift={(565.27,153.46)}, rotate = 539.79] [color={rgb, 255:red, 0; green, 0; blue, 0 }  ][line width=0.75]    (10.93,-3.29) .. controls (6.95,-1.4) and (3.31,-0.3) .. (0,0) .. controls (3.31,0.3) and (6.95,1.4) .. (10.93,3.29)   ;
\draw    (432.18,169.44) -- (479.25,237.22) ;
\draw [shift={(480.39,238.86)}, rotate = 235.22] [color={rgb, 255:red, 0; green, 0; blue, 0 }  ][line width=0.75]    (10.93,-3.29) .. controls (6.95,-1.4) and (3.31,-0.3) .. (0,0) .. controls (3.31,0.3) and (6.95,1.4) .. (10.93,3.29)   ;
\draw    (527.92,238.86) -- (587.72,170.55) ;
\draw [shift={(589.03,169.04)}, rotate = 491.2] [color={rgb, 255:red, 0; green, 0; blue, 0 }  ][line width=0.75]    (10.93,-3.29) .. controls (6.95,-1.4) and (3.31,-0.3) .. (0,0) .. controls (3.31,0.3) and (6.95,1.4) .. (10.93,3.29)   ;
\draw    (455.94,405.13) -- (566.65,404.64) ;
\draw [shift={(568.65,404.63)}, rotate = 539.74] [color={rgb, 255:red, 0; green, 0; blue, 0 }  ][line width=0.75]    (10.93,-3.29) .. controls (6.95,-1.4) and (3.31,-0.3) .. (0,0) .. controls (3.31,0.3) and (6.95,1.4) .. (10.93,3.29)   ;
\draw    (528,304.58) -- (584.96,388.58) ;
\draw [shift={(586.08,390.24)}, rotate = 235.86] [color={rgb, 255:red, 0; green, 0; blue, 0 }  ][line width=0.75]    (10.93,-3.29) .. controls (6.95,-1.4) and (3.31,-0.3) .. (0,0) .. controls (3.31,0.3) and (6.95,1.4) .. (10.93,3.29)   ;
\draw    (432.18,389.55) -- (479.49,306.32) ;
\draw [shift={(480.48,304.58)}, rotate = 479.61] [color={rgb, 255:red, 0; green, 0; blue, 0 }  ][line width=0.75]    (10.93,-3.29) .. controls (6.95,-1.4) and (3.31,-0.3) .. (0,0) .. controls (3.31,0.3) and (6.95,1.4) .. (10.93,3.29)   ;
\draw    (527.92,238.86) .. controls (607.4,259.17) and (664.54,379.37) .. (617.63,403.92) ;
\draw [shift={(616.18,404.63)}, rotate = 335.26] [color={rgb, 255:red, 0; green, 0; blue, 0 }  ][line width=0.75]    (10.93,-3.29) .. controls (6.95,-1.4) and (3.31,-0.3) .. (0,0) .. controls (3.31,0.3) and (6.95,1.4) .. (10.93,3.29)   ;
\draw    (528,304.58) .. controls (565.04,296.2) and (677.8,209.43) .. (613.77,154.29) ;
\draw [shift={(612.8,153.46)}, rotate = 399.78999999999996] [color={rgb, 255:red, 0; green, 0; blue, 0 }  ][line width=0.75]    (10.93,-3.29) .. controls (6.95,-1.4) and (3.31,-0.3) .. (0,0) .. controls (3.31,0.3) and (6.95,1.4) .. (10.93,3.29)   ;
\draw    (592.41,420.21) .. controls (563.87,460.87) and (413.8,525.13) .. (99.38,476.43) ;
\draw [shift={(99.38,476.43)}, rotate = 368.8] [color={rgb, 255:red, 0; green, 0; blue, 0 }  ][line width=0.75]    (10.93,-3.29) .. controls (6.95,-1.4) and (3.31,-0.3) .. (0,0) .. controls (3.31,0.3) and (6.95,1.4) .. (10.93,3.29)   ;
\draw   (307.48,157.58) .. controls (307.48,148.98) and (318.11,142) .. (331.24,142) .. controls (344.36,142) and (355,148.98) .. (355,157.58) .. controls (355,166.18) and (344.36,173.16) .. (331.24,173.16) .. controls (318.11,173.16) and (307.48,166.18) .. (307.48,157.58) -- cycle ;
\draw    (589.03,137.88) .. controls (636.4,108.64) and (206.81,29.87) .. (39.24,142) ;
\draw [shift={(39.24,142)}, rotate = 326.21000000000004] [color={rgb, 255:red, 0; green, 0; blue, 0 }  ][line width=0.75]    (10.93,-3.29) .. controls (6.95,-1.4) and (3.31,-0.3) .. (0,0) .. controls (3.31,0.3) and (6.95,1.4) .. (10.93,3.29)   ;
\draw   (142.46,159.17) .. controls (142.46,150.56) and (153.1,143.59) .. (166.22,143.59) .. controls (179.34,143.59) and (189.98,150.56) .. (189.98,159.17) .. controls (189.98,167.77) and (179.34,174.74) .. (166.22,174.74) .. controls (153.1,174.74) and (142.46,167.77) .. (142.46,159.17) -- cycle ;
\draw   (217.48,244.58) .. controls (217.48,235.98) and (228.11,229) .. (241.24,229) .. controls (254.36,229) and (265,235.98) .. (265,244.58) .. controls (265,253.18) and (254.36,260.16) .. (241.24,260.16) .. controls (228.11,260.16) and (217.48,253.18) .. (217.48,244.58) -- cycle ;
\draw   (220.61,306.71) .. controls (220.61,298.1) and (231.25,291.13) .. (244.37,291.13) .. controls (257.49,291.13) and (268.13,298.1) .. (268.13,306.71) .. controls (268.13,315.31) and (257.49,322.29) .. (244.37,322.29) .. controls (231.25,322.29) and (220.61,315.31) .. (220.61,306.71) -- cycle ;
\draw   (154.54,404.63) .. controls (154.54,396.02) and (165.18,389.05) .. (178.3,389.05) .. controls (191.43,389.05) and (202.06,396.02) .. (202.06,404.63) .. controls (202.06,413.23) and (191.43,420.21) .. (178.3,420.21) .. controls (165.18,420.21) and (154.54,413.23) .. (154.54,404.63) -- cycle ;
\draw   (316.41,405.43) .. controls (316.41,396.82) and (327.04,389.85) .. (340.17,389.85) .. controls (353.29,389.85) and (363.93,396.82) .. (363.93,405.43) .. controls (363.93,414.03) and (353.29,421.01) .. (340.17,421.01) .. controls (327.04,421.01) and (316.41,414.03) .. (316.41,405.43) -- cycle ;
\draw   (408.42,405.13) .. controls (408.42,396.53) and (419.06,389.55) .. (432.18,389.55) .. controls (445.31,389.55) and (455.94,396.53) .. (455.94,405.13) .. controls (455.94,413.74) and (445.31,420.71) .. (432.18,420.71) .. controls (419.06,420.71) and (408.42,413.74) .. (408.42,405.13) -- cycle ;
\draw   (480.48,304.58) .. controls (480.48,295.98) and (491.11,289) .. (504.24,289) .. controls (517.36,289) and (528,295.98) .. (528,304.58) .. controls (528,313.18) and (517.36,320.16) .. (504.24,320.16) .. controls (491.11,320.16) and (480.48,313.18) .. (480.48,304.58) -- cycle ;
\draw   (480.39,238.86) .. controls (480.39,230.26) and (491.03,223.28) .. (504.15,223.28) .. controls (517.28,223.28) and (527.92,230.26) .. (527.92,238.86) .. controls (527.92,247.47) and (517.28,254.44) .. (504.15,254.44) .. controls (491.03,254.44) and (480.39,247.47) .. (480.39,238.86) -- cycle ;
\draw   (565.27,153.46) .. controls (565.27,144.86) and (575.91,137.88) .. (589.03,137.88) .. controls (602.16,137.88) and (612.8,144.86) .. (612.8,153.46) .. controls (612.8,162.07) and (602.16,169.04) .. (589.03,169.04) .. controls (575.91,169.04) and (565.27,162.07) .. (565.27,153.46) -- cycle ;
\draw   (408.42,153.86) .. controls (408.42,145.26) and (419.06,138.28) .. (432.18,138.28) .. controls (445.31,138.28) and (455.94,145.26) .. (455.94,153.86) .. controls (455.94,162.47) and (445.31,169.44) .. (432.18,169.44) .. controls (419.06,169.44) and (408.42,162.47) .. (408.42,153.86) -- cycle ;
\draw   (568.65,404.63) .. controls (568.65,396.03) and (579.29,389.05) .. (592.41,389.05) .. controls (605.54,389.05) and (616.18,396.03) .. (616.18,404.63) .. controls (616.18,413.23) and (605.54,420.21) .. (592.41,420.21) .. controls (579.29,420.21) and (568.65,413.23) .. (568.65,404.63) -- cycle ;
\draw   (51.85,476.43) .. controls (51.85,467.82) and (62.49,460.85) .. (75.62,460.85) .. controls (88.74,460.85) and (99.38,467.82) .. (99.38,476.43) .. controls (99.38,485.03) and (88.74,492.01) .. (75.62,492.01) .. controls (62.49,492.01) and (51.85,485.03) .. (51.85,476.43) -- cycle ;
\draw   (15.48,157.58) .. controls (15.48,148.98) and (26.11,142) .. (39.24,142) .. controls (52.36,142) and (63,148.98) .. (63,157.58) .. controls (63,166.18) and (52.36,173.16) .. (39.24,173.16) .. controls (26.11,173.16) and (15.48,166.18) .. (15.48,157.58) -- cycle ;

\draw (162.78,188.05) node [anchor=north west][inner sep=0.75pt]  [rotate=-53.84]  {$x_{i} =1$};
\draw (214.88,139.25) node [anchor=north west][inner sep=0.75pt]    {$x_{i} =0$};
\draw (162.88,352.66) node [anchor=north west][inner sep=0.75pt]  [rotate=-302]  {$x_{i} =1$};
\draw (236.86,383.12) node [anchor=north west][inner sep=0.75pt]    {$x_{i} =0$};
\draw (287.42,224.83) node [anchor=north west][inner sep=0.75pt]  [rotate=-314.83]  {$y_{i} =0$};
\draw (281.28,334.1) node [anchor=north west][inner sep=0.75pt]  [rotate=-49.84]  {$y_{i} =0$};
\draw (352.83,269.45) node [anchor=north west][inner sep=0.75pt]  [rotate=-301.95]  {$y_{i} =1$};
\draw (372.1,299.78) node [anchor=north west][inner sep=0.75pt]  [rotate=-50.08]  {$y_{i} =1$};
\draw (225.67,232.99) node [anchor=north west][inner sep=0.75pt]    {$b_{i}{}_{,1}$};
\draw (232.22,296.42) node [anchor=north west][inner sep=0.75pt]    {$b_{i,}{}_{2}$};
\draw (166.03,395.29) node [anchor=north west][inner sep=0.75pt]    {$a_{i,}{}_{2}$};
\draw (156.47,145.51) node [anchor=north west][inner sep=0.75pt]    {$a_{i}{}_{,1}$};
\draw (430.22,193.21) node [anchor=north west][inner sep=0.75pt]  [rotate=-65.82]  {$x_{n} =1$};
\draw (477.34,129.22) node [anchor=north west][inner sep=0.75pt]    {$x_{n} =0$};
\draw (425.53,357.37) node [anchor=north west][inner sep=0.75pt]  [rotate=-296.24]  {$x_{n} =1$};
\draw (478.44,384.09) node [anchor=north west][inner sep=0.75pt]    {$x_{n} =0$};
\draw (542.33,225.24) node [anchor=north west][inner sep=0.75pt]  [rotate=-313.52]  {$y_{n} =0$};
\draw (572.08,323.73) node [anchor=north west][inner sep=0.75pt]  [rotate=-59.32]  {$y_{n} =0$};
\draw (608.93,254.75) node [anchor=north west][inner sep=0.75pt]  [rotate=-303.05,xslant=0.02]  {$y_{n} =1$};
\draw (626.85,284.8) node [anchor=north west][inner sep=0.75pt]  [rotate=-62.79]  {$y_{n} =1$};
\draw (488.14,226.06) node [anchor=north west][inner sep=0.75pt]    {$b_{n,}{}_{1}$};
\draw (489.91,293.72) node [anchor=north west][inner sep=0.75pt]    {$b_{n,}{}_{2}$};
\draw (417.92,395.31) node [anchor=north west][inner sep=0.75pt]    {$a_{n,}{}_{2}$};
\draw (414.49,143.31) node [anchor=north west][inner sep=0.75pt]    {$a_{n,}{}_{1}$};
\draw (367.69,397.33) node [anchor=north west][inner sep=0.75pt]    {$\cdots \cdots $};
\draw (361.48,142.07) node [anchor=north west][inner sep=0.75pt]    {$\cdots \cdots $};
\draw (123.03,421.99) node [anchor=north west][inner sep=0.75pt]    {$i=1$};
\draw (566.54,141.75) node [anchor=north west][inner sep=0.75pt]    {$a_{n}{}_{+}{}_{1}{}_{,}{}_{1}$};
\draw (74.21,467.93) node [anchor=north west][inner sep=0.75pt]    {$r$};
\draw (133.38,553.24) node [anchor=north west][inner sep=0.75pt]   [align=left] {\text{Figure} 1 :};
\draw (572.23,390.36) node [anchor=north west][inner sep=0.75pt]    {$b_{n}{}_{+}{}_{1}{}_{,}{}_{1}$};
\draw (24.95,146.86) node [anchor=north west][inner sep=0.75pt]    {$a_{1}{}_{,}{}_{1}$};
\draw (64.1,148.83) node [anchor=north west][inner sep=0.75pt]    {$\cdots \cdots \cdots $};
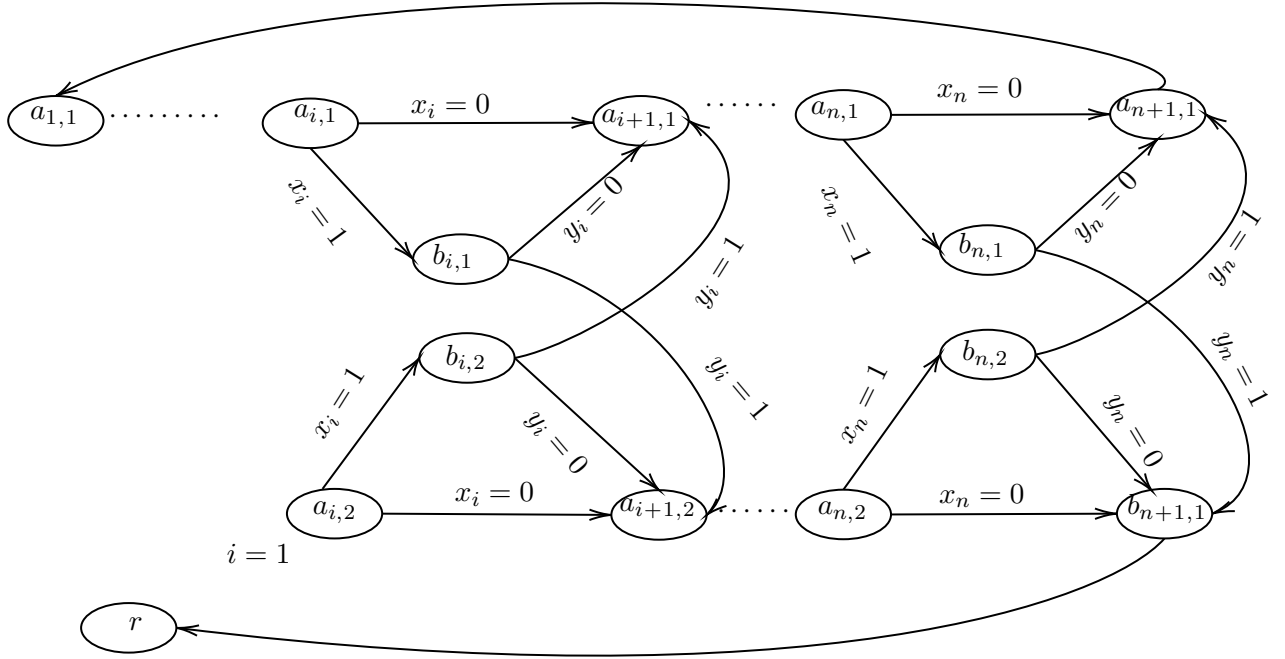
\draw (199.35,552.45) node [anchor=north west][inner sep=0.75pt]    {$\text{Sub-gadgets}\ {G}_{i} \ for\ 1\leq i\leq \ n+1.$};
\draw (310,148.4) node [anchor=north west][inner sep=0.75pt]    {$a_{i}{}_{+}{}_{1}{}_{,}{}_{1}$};
\draw (319,393.4) node [anchor=north west][inner sep=0.75pt]    {$a_{i}{}_{+}{}_{1}{}_{,}{}_{2}$};
\end{tikzpicture}
\\\newline
In any $G_i$ the edges are oriented in the following way. When $x_i=0$, $\tau_x(a_{i,1})=a_{i+1,1}$ and $\tau_x(a_{i,2})=a_{i+1,2}$. For $x_i=1$, $\tau_x(a_{i,1})=b_{i,1}$ and $\tau_x(a_{i,2})=b_{i,2}$. On the other hand, when $y_i=0$, $\theta_y(b_{i,1})=a_{i+1,1}$ and $\theta_y(b_{i,2})=a_{i+1,2}$. When $y_i=1$, $\theta_y(b_{i,1})=a_{i+1,2}$ and $\theta_y(b_{i,2})=a_{i+1,1}$. In sub-gadget $G_{n+1}$, $\tau_x(a_{n+1,1})=a_{1,1}$ and $\theta_y(b_{n+1,1})=r$.



\begin{claim}
\label{rectangular-reduction}
Given $(x,y)\in\{0,1\}^n\times\{0,1\}^n$ the reduction $\red$ computes the graph $G_{x,y}$ such that the oriented edges form a spanning tree rooted at $r$ if and only if 
$\angle{x,y}=1 \mod 2$. 
\end{claim}

First we give two illustrative examples of the reduction $\red$ for $n=3$ in Figure $2$ and Figure $3$. In the first example the input is $x=(1,1,1)$ and $y=(1,1,1)$. In this case we get a spanning tree rooted at $r$ shown by blue edges in Figure $2$. In the second example the input is $x=(0,1,1)$ and $y=(1,1,1)$. In this case we get a cycle shown by the red edges in Figure $3$.

\tikzset{every picture/.style={line width=0.75pt}} 

\begin{tikzpicture}[x=0.75pt,y=0.75pt,yscale=-1,xscale=1]

\draw   (275.99,68.27) .. controls (275.99,58.75) and (284.04,51.04) .. (293.96,51.04) .. controls (303.88,51.04) and (311.92,58.75) .. (311.92,68.27) .. controls (311.92,77.79) and (303.88,85.51) .. (293.96,85.51) .. controls (284.04,85.51) and (275.99,77.79) .. (275.99,68.27) -- cycle ;
\draw   (277.68,311.22) .. controls (277.68,301.74) and (285.69,294.05) .. (295.57,294.05) .. controls (305.45,294.05) and (313.46,301.74) .. (313.46,311.22) .. controls (313.46,320.7) and (305.45,328.39) .. (295.57,328.39) .. controls (285.69,328.39) and (277.68,320.7) .. (277.68,311.22) -- cycle ;
\draw   (344.94,164.85) .. controls (344.94,155.79) and (352.6,148.44) .. (362.05,148.44) .. controls (371.49,148.44) and (379.15,155.79) .. (379.15,164.85) .. controls (379.15,173.92) and (371.49,181.27) .. (362.05,181.27) .. controls (352.6,181.27) and (344.94,173.92) .. (344.94,164.85) -- cycle ;
\draw   (411.7,68.01) .. controls (411.7,58.49) and (419.75,50.77) .. (429.67,50.77) .. controls (439.59,50.77) and (447.63,58.49) .. (447.63,68.01) .. controls (447.63,77.53) and (439.59,85.24) .. (429.67,85.24) .. controls (419.75,85.24) and (411.7,77.53) .. (411.7,68.01) -- cycle ;
\draw   (413.25,311.22) .. controls (413.25,301.7) and (421.29,293.99) .. (431.21,293.99) .. controls (441.13,293.99) and (449.17,301.7) .. (449.17,311.22) .. controls (449.17,320.74) and (441.13,328.46) .. (431.21,328.46) .. controls (421.29,328.46) and (413.25,320.74) .. (413.25,311.22) -- cycle ;
\draw   (479.95,165.32) .. controls (479.95,155.8) and (488,148.08) .. (497.92,148.08) .. controls (507.84,148.08) and (515.88,155.8) .. (515.88,165.32) .. controls (515.88,174.84) and (507.84,182.55) .. (497.92,182.55) .. controls (488,182.55) and (479.95,174.84) .. (479.95,165.32) -- cycle ;
\draw   (345.39,213.69) .. controls (345.39,204.17) and (353.43,196.45) .. (363.35,196.45) .. controls (373.27,196.45) and (381.32,204.17) .. (381.32,213.69) .. controls (381.32,223.21) and (373.27,230.92) .. (363.35,230.92) .. controls (353.43,230.92) and (345.39,223.21) .. (345.39,213.69) -- cycle ;
\draw   (481.1,213.69) .. controls (481.1,204.17) and (489.14,196.45) .. (499.06,196.45) .. controls (508.99,196.45) and (517.03,204.17) .. (517.03,213.69) .. controls (517.03,223.21) and (508.99,230.92) .. (499.06,230.92) .. controls (489.14,230.92) and (481.1,223.21) .. (481.1,213.69) -- cycle ;
\draw   (548.96,311.22) .. controls (548.96,301.7) and (557,293.99) .. (566.92,293.99) .. controls (576.84,293.99) and (584.88,301.7) .. (584.88,311.22) .. controls (584.88,320.74) and (576.84,328.46) .. (566.92,328.46) .. controls (557,328.46) and (548.96,320.74) .. (548.96,311.22) -- cycle ;
\draw   (547.41,67.55) .. controls (547.41,58.03) and (555.46,50.32) .. (565.38,50.32) .. controls (575.3,50.32) and (583.34,58.03) .. (583.34,67.55) .. controls (583.34,77.07) and (575.3,84.79) .. (565.38,84.79) .. controls (555.46,84.79) and (547.41,77.07) .. (547.41,67.55) -- cycle ;
\draw   (209.23,165.12) .. controls (209.23,156.05) and (216.89,148.7) .. (226.34,148.7) .. controls (235.78,148.7) and (243.44,156.05) .. (243.44,165.12) .. controls (243.44,174.19) and (235.78,181.54) .. (226.34,181.54) .. controls (216.89,181.54) and (209.23,174.19) .. (209.23,165.12) -- cycle ;
\draw   (142,68.54) .. controls (142,59.47) and (149.66,52.13) .. (159.1,52.13) .. controls (168.55,52.13) and (176.21,59.47) .. (176.21,68.54) .. controls (176.21,77.61) and (168.55,84.96) .. (159.1,84.96) .. controls (149.66,84.96) and (142,77.61) .. (142,68.54) -- cycle ;
\draw   (142.89,310.58) .. controls (142.89,301.52) and (150.55,294.17) .. (160,294.17) .. controls (169.45,294.17) and (177.11,301.52) .. (177.11,310.58) .. controls (177.11,319.65) and (169.45,327) .. (160,327) .. controls (150.55,327) and (142.89,319.65) .. (142.89,310.58) -- cycle ;
\draw [color={rgb, 255:red, 0; green, 0; blue, 255 }  ,draw opacity=1 ]   (379.15,164.85) .. controls (473.87,178.4) and (456.66,246.91) .. (431.96,292.61) ;
\draw [shift={(431.21,293.99)}, rotate = 298.77] [color={rgb, 255:red, 0; green, 0; blue, 255 }  ,draw opacity=1 ][line width=0.75]    (10.93,-3.29) .. controls (6.95,-1.4) and (3.31,-0.3) .. (0,0) .. controls (3.31,0.3) and (6.95,1.4) .. (10.93,3.29)   ;
\draw [color={rgb, 255:red, 0; green, 0; blue, 255 }  ,draw opacity=1 ]   (363.35,196.45) .. controls (475.85,192.41) and (454.14,120.55) .. (430.74,86.76) ;
\draw [shift={(429.67,85.24)}, rotate = 414.31] [color={rgb, 255:red, 0; green, 0; blue, 255 }  ,draw opacity=1 ][line width=0.75]    (10.93,-3.29) .. controls (6.95,-1.4) and (3.31,-0.3) .. (0,0) .. controls (3.31,0.3) and (6.95,1.4) .. (10.93,3.29)   ;
\draw [color={rgb, 255:red, 0; green, 0; blue, 255 }  ,draw opacity=1 ]   (311.92,68.27) -- (360.99,146.74) ;
\draw [shift={(362.05,148.44)}, rotate = 237.98] [color={rgb, 255:red, 0; green, 0; blue, 255 }  ,draw opacity=1 ][line width=0.75]    (10.93,-3.29) .. controls (6.95,-1.4) and (3.31,-0.3) .. (0,0) .. controls (3.31,0.3) and (6.95,1.4) .. (10.93,3.29)   ;
\draw [color={rgb, 255:red, 0; green, 0; blue, 255 }  ,draw opacity=1 ]   (447.63,68.01) -- (496.7,146.47) ;
\draw [shift={(497.76,148.17)}, rotate = 237.98] [color={rgb, 255:red, 0; green, 0; blue, 255 }  ,draw opacity=1 ][line width=0.75]    (10.93,-3.29) .. controls (6.95,-1.4) and (3.31,-0.3) .. (0,0) .. controls (3.31,0.3) and (6.95,1.4) .. (10.93,3.29)   ;
\draw [color={rgb, 255:red, 0; green, 0; blue, 255 }  ,draw opacity=1 ]   (313.46,311.22) -- (362.3,232.62) ;
\draw [shift={(363.35,230.92)}, rotate = 481.85] [color={rgb, 255:red, 0; green, 0; blue, 255 }  ,draw opacity=1 ][line width=0.75]    (10.93,-3.29) .. controls (6.95,-1.4) and (3.31,-0.3) .. (0,0) .. controls (3.31,0.3) and (6.95,1.4) .. (10.93,3.29)   ;
\draw [color={rgb, 255:red, 0; green, 0; blue, 255 }  ,draw opacity=1 ]   (449.17,311.22) -- (498.01,232.62) ;
\draw [shift={(499.06,230.92)}, rotate = 481.85] [color={rgb, 255:red, 0; green, 0; blue, 255 }  ,draw opacity=1 ][line width=0.75]    (10.93,-3.29) .. controls (6.95,-1.4) and (3.31,-0.3) .. (0,0) .. controls (3.31,0.3) and (6.95,1.4) .. (10.93,3.29)   ;
\draw [color={rgb, 255:red, 0; green, 0; blue, 255 }  ,draw opacity=1 ]   (515.88,165.32) .. controls (610.6,178.86) and (593.38,247.37) .. (568.69,293.07) ;
\draw [shift={(567.94,294.45)}, rotate = 298.77] [color={rgb, 255:red, 0; green, 0; blue, 255 }  ,draw opacity=1 ][line width=0.75]    (10.93,-3.29) .. controls (6.95,-1.4) and (3.31,-0.3) .. (0,0) .. controls (3.31,0.3) and (6.95,1.4) .. (10.93,3.29)   ;
\draw [color={rgb, 255:red, 0; green, 0; blue, 255 }  ,draw opacity=1 ]   (517.03,213.69) .. controls (613.53,183.46) and (589.9,119.84) .. (566.45,86.3) ;
\draw [shift={(565.38,84.79)}, rotate = 414.31] [color={rgb, 255:red, 0; green, 0; blue, 255 }  ,draw opacity=1 ][line width=0.75]    (10.93,-3.29) .. controls (6.95,-1.4) and (3.31,-0.3) .. (0,0) .. controls (3.31,0.3) and (6.95,1.4) .. (10.93,3.29)   ;
\draw   (299.21,363.07) .. controls (299.21,354.3) and (306.61,347.2) .. (315.74,347.2) .. controls (324.88,347.2) and (332.28,354.3) .. (332.28,363.07) .. controls (332.28,371.83) and (324.88,378.94) .. (315.74,378.94) .. controls (306.61,378.94) and (299.21,371.83) .. (299.21,363.07) -- cycle ;
\draw [color={rgb, 255:red, 0; green, 0; blue, 255 }  ,draw opacity=1 ]   (566.92,328.46) .. controls (537.7,383.45) and (395.01,405.39) .. (333.2,363.7) ;
\draw [shift={(332.28,363.07)}, rotate = 394.76] [color={rgb, 255:red, 0; green, 0; blue, 255 }  ,draw opacity=1 ][line width=0.75]    (10.93,-3.29) .. controls (6.95,-1.4) and (3.31,-0.3) .. (0,0) .. controls (3.31,0.3) and (6.95,1.4) .. (10.93,3.29)   ;
\draw [color={rgb, 255:red, 0; green, 0; blue, 255 }  ,draw opacity=1 ]   (565.38,50.32) .. controls (531.85,21.22) and (339.41,-19.82) .. (159.1,52.13) ;
\draw [shift={(159.1,52.13)}, rotate = 338.25] [color={rgb, 255:red, 0; green, 0; blue, 255 }  ,draw opacity=1 ][line width=0.75]    (10.93,-3.29) .. controls (6.95,-1.4) and (3.31,-0.3) .. (0,0) .. controls (3.31,0.3) and (6.95,1.4) .. (10.93,3.29)   ;
\draw    (275.99,68.27) -- (226.34,148.7) ;
\draw    (311.92,68.27) -- (411.7,68.01) ;
\draw    (411.7,68.01) -- (362.05,148.44) ;
\draw    (363.35,230.92) -- (413.25,311.22) ;
\draw    (313.46,311.22) -- (413.25,311.22) ;
\draw    (447.63,68.01) -- (547.41,67.55) ;
\draw    (547.41,67.55) -- (497.76,148.17) ;
\draw    (499.06,230.92) -- (548.96,311.22) ;
\draw    (449.17,311.22) -- (548.96,311.22) ;
\draw [color={rgb, 255:red, 0; green, 0; blue, 255 }  ,draw opacity=1 ]   (311.92,68.27) -- (360.99,146.74) ;
\draw [shift={(362.05,148.44)}, rotate = 237.98] [color={rgb, 255:red, 0; green, 0; blue, 255 }  ,draw opacity=1 ][line width=0.75]    (10.93,-3.29) .. controls (6.95,-1.4) and (3.31,-0.3) .. (0,0) .. controls (3.31,0.3) and (6.95,1.4) .. (10.93,3.29)   ;
\draw [color={rgb, 255:red, 0; green, 0; blue, 255 }  ,draw opacity=1 ]   (379.15,164.85) .. controls (473.87,178.4) and (456.66,246.91) .. (431.96,292.61) ;
\draw [shift={(431.21,293.99)}, rotate = 298.77] [color={rgb, 255:red, 0; green, 0; blue, 255 }  ,draw opacity=1 ][line width=0.75]    (10.93,-3.29) .. controls (6.95,-1.4) and (3.31,-0.3) .. (0,0) .. controls (3.31,0.3) and (6.95,1.4) .. (10.93,3.29)   ;
\draw [color={rgb, 255:red, 0; green, 0; blue, 255 }  ,draw opacity=1 ]   (449.17,311.22) -- (498.01,232.62) ;
\draw [shift={(499.06,230.92)}, rotate = 481.85] [color={rgb, 255:red, 0; green, 0; blue, 255 }  ,draw opacity=1 ][line width=0.75]    (10.93,-3.29) .. controls (6.95,-1.4) and (3.31,-0.3) .. (0,0) .. controls (3.31,0.3) and (6.95,1.4) .. (10.93,3.29)   ;
\draw [color={rgb, 255:red, 0; green, 0; blue, 255 }  ,draw opacity=1 ]   (517.03,213.69) .. controls (613.53,183.46) and (589.9,119.84) .. (566.45,86.3) ;
\draw [shift={(565.38,84.79)}, rotate = 414.31] [color={rgb, 255:red, 0; green, 0; blue, 255 }  ,draw opacity=1 ][line width=0.75]    (10.93,-3.29) .. controls (6.95,-1.4) and (3.31,-0.3) .. (0,0) .. controls (3.31,0.3) and (6.95,1.4) .. (10.93,3.29)   ;
\draw [color={rgb, 255:red, 0; green, 0; blue, 255 }  ,draw opacity=1 ]   (565.38,50.32) .. controls (531.85,21.22) and (339.41,-19.82) .. (159.1,52.13) ;
\draw [shift={(159.1,52.13)}, rotate = 338.25] [color={rgb, 255:red, 0; green, 0; blue, 255 }  ,draw opacity=1 ][line width=0.75]    (10.93,-3.29) .. controls (6.95,-1.4) and (3.31,-0.3) .. (0,0) .. controls (3.31,0.3) and (6.95,1.4) .. (10.93,3.29)   ;
\draw   (211.23,217.12) .. controls (211.23,208.05) and (218.89,200.7) .. (228.34,200.7) .. controls (237.78,200.7) and (245.44,208.05) .. (245.44,217.12) .. controls (245.44,226.19) and (237.78,233.54) .. (228.34,233.54) .. controls (218.89,233.54) and (211.23,226.19) .. (211.23,217.12) -- cycle ;
\draw [color={rgb, 255:red, 0; green, 0; blue, 255 }  ,draw opacity=1 ]   (243.51,164.92) .. controls (338.23,178.46) and (321.02,246.98) .. (296.32,292.68) ;
\draw [shift={(295.57,294.05)}, rotate = 298.77] [color={rgb, 255:red, 0; green, 0; blue, 255 }  ,draw opacity=1 ][line width=0.75]    (10.93,-3.29) .. controls (6.95,-1.4) and (3.31,-0.3) .. (0,0) .. controls (3.31,0.3) and (6.95,1.4) .. (10.93,3.29)   ;
\draw [color={rgb, 255:red, 0; green, 0; blue, 255 }  ,draw opacity=1 ]   (228.34,200.7) .. controls (340.83,196.66) and (318.45,120.93) .. (295.03,87.03) ;
\draw [shift={(293.96,85.51)}, rotate = 414.31] [color={rgb, 255:red, 0; green, 0; blue, 255 }  ,draw opacity=1 ][line width=0.75]    (10.93,-3.29) .. controls (6.95,-1.4) and (3.31,-0.3) .. (0,0) .. controls (3.31,0.3) and (6.95,1.4) .. (10.93,3.29)   ;
\draw    (228.34,233.54) -- (277.68,311.22) ;
\draw    (176.21,68.54) -- (275.99,68.27) ;
\draw [color={rgb, 255:red, 0; green, 0; blue, 255 }  ,draw opacity=1 ]   (176.21,68.54) -- (225.27,147.01) ;
\draw [shift={(226.34,148.7)}, rotate = 237.98] [color={rgb, 255:red, 0; green, 0; blue, 255 }  ,draw opacity=1 ][line width=0.75]    (10.93,-3.29) .. controls (6.95,-1.4) and (3.31,-0.3) .. (0,0) .. controls (3.31,0.3) and (6.95,1.4) .. (10.93,3.29)   ;
\draw [color={rgb, 255:red, 0; green, 0; blue, 255 }  ,draw opacity=1 ]   (176.21,68.54) -- (225.27,147.01) ;
\draw [shift={(226.34,148.7)}, rotate = 237.98] [color={rgb, 255:red, 0; green, 0; blue, 255 }  ,draw opacity=1 ][line width=0.75]    (10.93,-3.29) .. controls (6.95,-1.4) and (3.31,-0.3) .. (0,0) .. controls (3.31,0.3) and (6.95,1.4) .. (10.93,3.29)   ;
\draw [color={rgb, 255:red, 0; green, 0; blue, 255 }  ,draw opacity=1 ]   (177.11,310.58) -- (225.94,231.99) ;
\draw [shift={(227,230.29)}, rotate = 481.85] [color={rgb, 255:red, 0; green, 0; blue, 255 }  ,draw opacity=1 ][line width=0.75]    (10.93,-3.29) .. controls (6.95,-1.4) and (3.31,-0.3) .. (0,0) .. controls (3.31,0.3) and (6.95,1.4) .. (10.93,3.29)   ;
\draw [color={rgb, 255:red, 0; green, 0; blue, 255 }  ,draw opacity=1 ]   (313.46,311.22) -- (362.3,232.62) ;
\draw [shift={(363.35,230.92)}, rotate = 481.85] [color={rgb, 255:red, 0; green, 0; blue, 255 }  ,draw opacity=1 ][line width=0.75]    (10.93,-3.29) .. controls (6.95,-1.4) and (3.31,-0.3) .. (0,0) .. controls (3.31,0.3) and (6.95,1.4) .. (10.93,3.29)   ;
\draw [color={rgb, 255:red, 0; green, 0; blue, 255 }  ,draw opacity=1 ]   (177.11,310.58) -- (225.94,231.99) ;
\draw [shift={(227,230.29)}, rotate = 481.85] [color={rgb, 255:red, 0; green, 0; blue, 255 }  ,draw opacity=1 ][line width=0.75]    (10.93,-3.29) .. controls (6.95,-1.4) and (3.31,-0.3) .. (0,0) .. controls (3.31,0.3) and (6.95,1.4) .. (10.93,3.29)   ;
\draw    (177.11,310.58) -- (276.89,310.58) ;
\draw [color={rgb, 255:red, 0; green, 0; blue, 255 }  ,draw opacity=1 ]   (243.51,164.92) .. controls (338.23,178.46) and (321.02,246.98) .. (296.32,292.68) ;
\draw [shift={(295.57,294.05)}, rotate = 298.77] [color={rgb, 255:red, 0; green, 0; blue, 255 }  ,draw opacity=1 ][line width=0.75]    (10.93,-3.29) .. controls (6.95,-1.4) and (3.31,-0.3) .. (0,0) .. controls (3.31,0.3) and (6.95,1.4) .. (10.93,3.29)   ;
\draw [color={rgb, 255:red, 0; green, 0; blue, 255 }  ,draw opacity=1 ]   (228.34,200.7) .. controls (340.83,196.66) and (318.45,120.93) .. (295.03,87.03) ;
\draw [shift={(293.96,85.51)}, rotate = 414.31] [color={rgb, 255:red, 0; green, 0; blue, 255 }  ,draw opacity=1 ][line width=0.75]    (10.93,-3.29) .. controls (6.95,-1.4) and (3.31,-0.3) .. (0,0) .. controls (3.31,0.3) and (6.95,1.4) .. (10.93,3.29)   ;
\draw [color={rgb, 255:red, 0; green, 0; blue, 255 }  ,draw opacity=1 ]   (363.35,196.45) .. controls (475.85,192.41) and (454.14,120.55) .. (430.74,86.76) ;
\draw [shift={(429.67,85.24)}, rotate = 414.31] [color={rgb, 255:red, 0; green, 0; blue, 255 }  ,draw opacity=1 ][line width=0.75]    (10.93,-3.29) .. controls (6.95,-1.4) and (3.31,-0.3) .. (0,0) .. controls (3.31,0.3) and (6.95,1.4) .. (10.93,3.29)   ;
\draw [color={rgb, 255:red, 0; green, 0; blue, 255 }  ,draw opacity=1 ]   (447.63,68.01) -- (496.7,146.47) ;
\draw [shift={(497.76,148.17)}, rotate = 237.98] [color={rgb, 255:red, 0; green, 0; blue, 255 }  ,draw opacity=1 ][line width=0.75]    (10.93,-3.29) .. controls (6.95,-1.4) and (3.31,-0.3) .. (0,0) .. controls (3.31,0.3) and (6.95,1.4) .. (10.93,3.29)   ;
\draw [color={rgb, 255:red, 0; green, 0; blue, 255 }  ,draw opacity=1 ]   (515.88,165.32) .. controls (610.6,178.86) and (593.38,247.37) .. (568.69,293.07) ;
\draw [shift={(567.94,294.45)}, rotate = 298.77] [color={rgb, 255:red, 0; green, 0; blue, 255 }  ,draw opacity=1 ][line width=0.75]    (10.93,-3.29) .. controls (6.95,-1.4) and (3.31,-0.3) .. (0,0) .. controls (3.31,0.3) and (6.95,1.4) .. (10.93,3.29)   ;
\draw [color={rgb, 255:red, 0; green, 0; blue, 255 }  ,draw opacity=1 ]   (566.92,328.46) .. controls (537.7,383.45) and (395.01,405.39) .. (333.2,363.7) ;
\draw [shift={(332.28,363.07)}, rotate = 394.76] [color={rgb, 255:red, 0; green, 0; blue, 255 }  ,draw opacity=1 ][line width=0.75]    (10.93,-3.29) .. controls (6.95,-1.4) and (3.31,-0.3) .. (0,0) .. controls (3.31,0.3) and (6.95,1.4) .. (10.93,3.29)   ;

\draw (281.21,60.36) node [anchor=north west][inner sep=0.75pt]    {$a_{2}{}_{,1}$};
\draw (145.36,53.99) node [anchor=north west][inner sep=0.75pt]    {$a_{1}{}_{,1}$};
\draw (421.2,57.36) node [anchor=north west][inner sep=0.75pt]    {$a_{3}{}_{,1}$};
\draw (551.35,56.7) node [anchor=north west][inner sep=0.75pt]    {$a_{4}{}_{,}{}_{1}$};
\draw (149.21,299.85) node [anchor=north west][inner sep=0.75pt]    {$a_{1}{}_{,2}$};
\draw (280.21,297.22) node [anchor=north west][inner sep=0.75pt]    {$a_{2}{}_{,2}$};
\draw (417.63,298.22) node [anchor=north west][inner sep=0.75pt]    {$a_{3}{}_{,2}$};
\draw (214.51,151.74) node [anchor=north west][inner sep=0.75pt]    {$b_{1}{}_{,1}$};
\draw (349.93,153.11) node [anchor=north west][inner sep=0.75pt]    {$b_{2}{}_{,1}$};
\draw (486.78,154.48) node [anchor=north west][inner sep=0.75pt]    {$b_{3}{}_{,1}$};
\draw (351.35,202.36) node [anchor=north west][inner sep=0.75pt]    {$b_{2}{}_{,2}$};
\draw (485.35,203.73) node [anchor=north west][inner sep=0.75pt]    {$b_{3}{}_{,2}$};
\draw (555.35,298.56) node [anchor=north west][inner sep=0.75pt]    {$b_{4}{}_{,}{}_{1}$};
\draw (310.53,354.65) node [anchor=north west][inner sep=0.75pt]    {$r$};
\draw (344.82,73.44) node [anchor=north west][inner sep=0.75pt]  [rotate=-59.58]  {$x_{2} =1$};
\draw (483.83,80.06) node [anchor=north west][inner sep=0.75pt]  [rotate=-62.33]  {$x_{3} =1$};
\draw (331.47,284.54) node [anchor=north west][inner sep=0.75pt]  [rotate=-304.79]  {$x_{2} =1$};
\draw (465.89,293.3) node [anchor=north west][inner sep=0.75pt]  [rotate=-303.17]  {$x_{3} =1$};
\draw (288.96,197.7) node [anchor=north west][inner sep=0.75pt]  [rotate=-64.47]  {$y_{1} =1$};
\draw (397.22,131.48) node [anchor=north west][inner sep=0.75pt]  [rotate=-28.38]  {$y_{2} =1$};
\draw (398.88,213.77) node [anchor=north west][inner sep=0.75pt]  [rotate=-325.25]  {$y_{2} =1$};
\draw (584.86,156.87) node [anchor=north west][inner sep=0.75pt]  [rotate=-284.7]  {$y_{3} =1$};
\draw (526.24,209.22) node [anchor=north west][inner sep=0.75pt]  [rotate=-339.22]  {$y_{3} =1$};
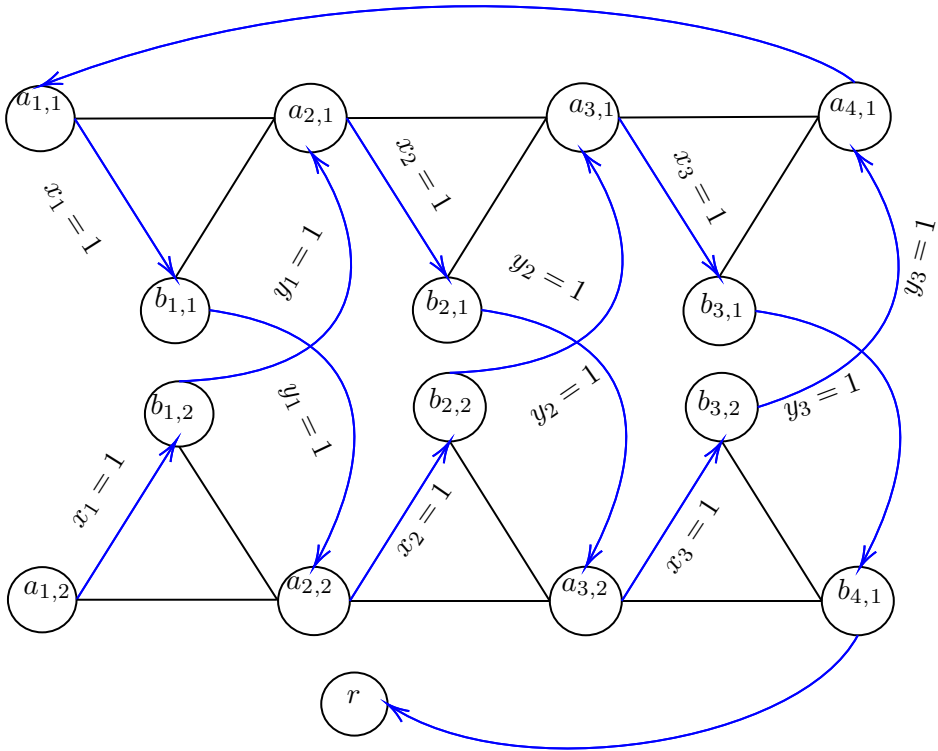
\draw (180.04,402.1) node [anchor=north west][inner sep=0.75pt]    {$\text{Figure}\ 2:\ \text{Gadget for}\ \sum _{i=1}^{3} x_{i} y_{i} =1\bmod 2.$};
\draw (168.95,95.88) node [anchor=north west][inner sep=0.75pt]  [rotate=-57.35]  {$x_{1} =1$};
\draw (212.51,205.74) node [anchor=north west][inner sep=0.75pt]    {$b_{1}{}_{,2}$};
\draw (270.1,155.55) node [anchor=north west][inner sep=0.75pt]  [rotate=-300.01]  {$y_{1} =1$};
\draw (169.56,270.41) node [anchor=north west][inner sep=0.75pt]  [rotate=-302.32]  {$x_{1} =1$};
\end{tikzpicture}
\newline
\newline

\tikzset{every picture/.style={line width=0.75pt}} 

\begin{tikzpicture}[x=0.75pt,y=0.75pt,yscale=-.90,xscale=1]

\draw   (275.99,68.27) .. controls (275.99,58.75) and (284.04,51.04) .. (293.96,51.04) .. controls (303.88,51.04) and (311.92,58.75) .. (311.92,68.27) .. controls (311.92,77.79) and (303.88,85.51) .. (293.96,85.51) .. controls (284.04,85.51) and (275.99,77.79) .. (275.99,68.27) -- cycle ;
\draw   (277.68,311.22) .. controls (277.68,301.74) and (285.69,294.05) .. (295.57,294.05) .. controls (305.45,294.05) and (313.46,301.74) .. (313.46,311.22) .. controls (313.46,320.7) and (305.45,328.39) .. (295.57,328.39) .. controls (285.69,328.39) and (277.68,320.7) .. (277.68,311.22) -- cycle ;
\draw   (344.94,164.85) .. controls (344.94,155.79) and (352.6,148.44) .. (362.05,148.44) .. controls (371.49,148.44) and (379.15,155.79) .. (379.15,164.85) .. controls (379.15,173.92) and (371.49,181.27) .. (362.05,181.27) .. controls (352.6,181.27) and (344.94,173.92) .. (344.94,164.85) -- cycle ;
\draw   (411.7,68.01) .. controls (411.7,58.49) and (419.75,50.77) .. (429.67,50.77) .. controls (439.59,50.77) and (447.63,58.49) .. (447.63,68.01) .. controls (447.63,77.53) and (439.59,85.24) .. (429.67,85.24) .. controls (419.75,85.24) and (411.7,77.53) .. (411.7,68.01) -- cycle ;
\draw   (413.25,311.22) .. controls (413.25,301.7) and (421.29,293.99) .. (431.21,293.99) .. controls (441.13,293.99) and (449.17,301.7) .. (449.17,311.22) .. controls (449.17,320.74) and (441.13,328.46) .. (431.21,328.46) .. controls (421.29,328.46) and (413.25,320.74) .. (413.25,311.22) -- cycle ;
\draw   (479.95,165.32) .. controls (479.95,155.8) and (488,148.08) .. (497.92,148.08) .. controls (507.84,148.08) and (515.88,155.8) .. (515.88,165.32) .. controls (515.88,174.84) and (507.84,182.55) .. (497.92,182.55) .. controls (488,182.55) and (479.95,174.84) .. (479.95,165.32) -- cycle ;
\draw   (345.39,213.69) .. controls (345.39,204.17) and (353.43,196.45) .. (363.35,196.45) .. controls (373.27,196.45) and (381.32,204.17) .. (381.32,213.69) .. controls (381.32,223.21) and (373.27,230.92) .. (363.35,230.92) .. controls (353.43,230.92) and (345.39,223.21) .. (345.39,213.69) -- cycle ;
\draw   (481.1,213.69) .. controls (481.1,204.17) and (489.14,196.45) .. (499.06,196.45) .. controls (508.99,196.45) and (517.03,204.17) .. (517.03,213.69) .. controls (517.03,223.21) and (508.99,230.92) .. (499.06,230.92) .. controls (489.14,230.92) and (481.1,223.21) .. (481.1,213.69) -- cycle ;
\draw   (548.96,311.22) .. controls (548.96,301.7) and (557,293.99) .. (566.92,293.99) .. controls (576.84,293.99) and (584.88,301.7) .. (584.88,311.22) .. controls (584.88,320.74) and (576.84,328.46) .. (566.92,328.46) .. controls (557,328.46) and (548.96,320.74) .. (548.96,311.22) -- cycle ;
\draw   (547.41,67.55) .. controls (547.41,58.03) and (555.46,50.32) .. (565.38,50.32) .. controls (575.3,50.32) and (583.34,58.03) .. (583.34,67.55) .. controls (583.34,77.07) and (575.3,84.79) .. (565.38,84.79) .. controls (555.46,84.79) and (547.41,77.07) .. (547.41,67.55) -- cycle ;
\draw   (209.23,165.12) .. controls (209.23,156.05) and (216.89,148.7) .. (226.34,148.7) .. controls (235.78,148.7) and (243.44,156.05) .. (243.44,165.12) .. controls (243.44,174.19) and (235.78,181.54) .. (226.34,181.54) .. controls (216.89,181.54) and (209.23,174.19) .. (209.23,165.12) -- cycle ;
\draw   (142,68.54) .. controls (142,59.47) and (149.66,52.13) .. (159.1,52.13) .. controls (168.55,52.13) and (176.21,59.47) .. (176.21,68.54) .. controls (176.21,77.61) and (168.55,84.96) .. (159.1,84.96) .. controls (149.66,84.96) and (142,77.61) .. (142,68.54) -- cycle ;
\draw   (142.89,310.58) .. controls (142.89,301.52) and (150.55,294.17) .. (160,294.17) .. controls (169.45,294.17) and (177.11,301.52) .. (177.11,310.58) .. controls (177.11,319.65) and (169.45,327) .. (160,327) .. controls (150.55,327) and (142.89,319.65) .. (142.89,310.58) -- cycle ;
\draw [color={rgb, 255:red, 255; green, 0; blue, 0 }  ,draw opacity=1 ]   (379.15,164.85) .. controls (473.87,178.4) and (456.66,246.91) .. (431.96,292.61) ;
\draw [shift={(431.21,293.99)}, rotate = 298.77] [color={rgb, 255:red, 255; green, 0; blue, 0 }  ,draw opacity=1 ][line width=0.75]    (10.93,-3.29) .. controls (6.95,-1.4) and (3.31,-0.3) .. (0,0) .. controls (3.31,0.3) and (6.95,1.4) .. (10.93,3.29)   ;
\draw [color={rgb, 255:red, 0; green, 0; blue, 255 }  ,draw opacity=1 ]   (363.35,196.45) .. controls (475.85,192.41) and (454.14,120.55) .. (430.74,86.76) ;
\draw [shift={(429.67,85.24)}, rotate = 414.31] [color={rgb, 255:red, 0; green, 0; blue, 255 }  ,draw opacity=1 ][line width=0.75]    (10.93,-3.29) .. controls (6.95,-1.4) and (3.31,-0.3) .. (0,0) .. controls (3.31,0.3) and (6.95,1.4) .. (10.93,3.29)   ;
\draw [color={rgb, 255:red, 255; green, 0; blue, 0 }  ,draw opacity=1 ]   (311.92,68.27) -- (360.99,146.74) ;
\draw [shift={(362.05,148.44)}, rotate = 237.98] [color={rgb, 255:red, 255; green, 0; blue, 0 }  ,draw opacity=1 ][line width=0.75]    (10.93,-3.29) .. controls (6.95,-1.4) and (3.31,-0.3) .. (0,0) .. controls (3.31,0.3) and (6.95,1.4) .. (10.93,3.29)   ;
\draw [color={rgb, 255:red, 0; green, 0; blue, 255 }  ,draw opacity=1 ]   (447.63,68.01) -- (496.7,146.47) ;
\draw [shift={(497.76,148.17)}, rotate = 237.98] [color={rgb, 255:red, 0; green, 0; blue, 255 }  ,draw opacity=1 ][line width=0.75]    (10.93,-3.29) .. controls (6.95,-1.4) and (3.31,-0.3) .. (0,0) .. controls (3.31,0.3) and (6.95,1.4) .. (10.93,3.29)   ;
\draw [color={rgb, 255:red, 0; green, 0; blue, 255 }  ,draw opacity=1 ]   (313.46,311.22) -- (362.3,232.62) ;
\draw [shift={(363.35,230.92)}, rotate = 481.85] [color={rgb, 255:red, 0; green, 0; blue, 255 }  ,draw opacity=1 ][line width=0.75]    (10.93,-3.29) .. controls (6.95,-1.4) and (3.31,-0.3) .. (0,0) .. controls (3.31,0.3) and (6.95,1.4) .. (10.93,3.29)   ;
\draw [color={rgb, 255:red, 255; green, 0; blue, 0 }  ,draw opacity=1 ]   (449.17,311.22) -- (498.01,232.62) ;
\draw [shift={(499.06,230.92)}, rotate = 481.85] [color={rgb, 255:red, 255; green, 0; blue, 0 }  ,draw opacity=1 ][line width=0.75]    (10.93,-3.29) .. controls (6.95,-1.4) and (3.31,-0.3) .. (0,0) .. controls (3.31,0.3) and (6.95,1.4) .. (10.93,3.29)   ;
\draw [color={rgb, 255:red, 0; green, 0; blue, 255 }  ,draw opacity=1 ]   (515.88,165.32) .. controls (610.6,178.86) and (593.38,247.37) .. (568.69,293.07) ;
\draw [shift={(567.94,294.45)}, rotate = 298.77] [color={rgb, 255:red, 0; green, 0; blue, 255 }  ,draw opacity=1 ][line width=0.75]    (10.93,-3.29) .. controls (6.95,-1.4) and (3.31,-0.3) .. (0,0) .. controls (3.31,0.3) and (6.95,1.4) .. (10.93,3.29)   ;
\draw [color={rgb, 255:red, 255; green, 0; blue, 0 }  ,draw opacity=1 ]   (517.03,213.69) .. controls (613.53,183.46) and (589.9,119.84) .. (566.45,86.3) ;
\draw [shift={(565.38,84.79)}, rotate = 414.31] [color={rgb, 255:red, 255; green, 0; blue, 0 }  ,draw opacity=1 ][line width=0.75]    (10.93,-3.29) .. controls (6.95,-1.4) and (3.31,-0.3) .. (0,0) .. controls (3.31,0.3) and (6.95,1.4) .. (10.93,3.29)   ;
\draw   (299.21,363.07) .. controls (299.21,354.3) and (306.61,347.2) .. (315.74,347.2) .. controls (324.88,347.2) and (332.28,354.3) .. (332.28,363.07) .. controls (332.28,371.83) and (324.88,378.94) .. (315.74,378.94) .. controls (306.61,378.94) and (299.21,371.83) .. (299.21,363.07) -- cycle ;
\draw [color={rgb, 255:red, 0; green, 0; blue, 255 }  ,draw opacity=1 ]   (566.92,328.46) .. controls (537.7,383.45) and (395.01,405.39) .. (333.2,363.7) ;
\draw [shift={(332.28,363.07)}, rotate = 394.76] [color={rgb, 255:red, 0; green, 0; blue, 255 }  ,draw opacity=1 ][line width=0.75]    (10.93,-3.29) .. controls (6.95,-1.4) and (3.31,-0.3) .. (0,0) .. controls (3.31,0.3) and (6.95,1.4) .. (10.93,3.29)   ;
\draw [color={rgb, 255:red, 255; green, 0; blue, 0 }  ,draw opacity=1 ]   (565.38,50.32) .. controls (531.85,21.22) and (339.41,-19.82) .. (159.1,52.13) ;
\draw [shift={(159.1,52.13)}, rotate = 338.25] [color={rgb, 255:red, 255; green, 0; blue, 0 }  ,draw opacity=1 ][line width=0.75]    (10.93,-3.29) .. controls (6.95,-1.4) and (3.31,-0.3) .. (0,0) .. controls (3.31,0.3) and (6.95,1.4) .. (10.93,3.29)   ;
\draw [color={rgb, 255:red, 255; green, 0; blue, 0 }  ,draw opacity=1 ]   (176.21,68.54) -- (273.99,68.28) ;
\draw [shift={(275.99,68.27)}, rotate = 539.85] [color={rgb, 255:red, 255; green, 0; blue, 0 }  ,draw opacity=1 ][line width=0.75]    (10.93,-3.29) .. controls (6.95,-1.4) and (3.31,-0.3) .. (0,0) .. controls (3.31,0.3) and (6.95,1.4) .. (10.93,3.29)   ;
\draw    (176.21,68.54) -- (226.34,148.7) ;
\draw    (275.99,68.27) -- (226.34,148.7) ;
\draw    (311.92,68.27) -- (411.7,68.01) ;
\draw    (411.7,68.01) -- (362.05,148.44) ;
\draw    (363.35,230.92) -- (413.25,311.22) ;
\draw    (313.46,311.22) -- (413.25,311.22) ;
\draw    (447.63,68.01) -- (547.41,67.55) ;
\draw    (547.41,67.55) -- (497.76,148.17) ;
\draw    (499.06,230.92) -- (548.96,311.22) ;
\draw    (449.17,311.22) -- (548.96,311.22) ;
\draw [color={rgb, 255:red, 255; green, 0; blue, 0 }  ,draw opacity=1 ]   (176.21,68.54) -- (273.99,68.28) ;
\draw [shift={(275.99,68.27)}, rotate = 539.85] [color={rgb, 255:red, 255; green, 0; blue, 0 }  ,draw opacity=1 ][line width=0.75]    (10.93,-3.29) .. controls (6.95,-1.4) and (3.31,-0.3) .. (0,0) .. controls (3.31,0.3) and (6.95,1.4) .. (10.93,3.29)   ;
\draw [color={rgb, 255:red, 255; green, 0; blue, 0 }  ,draw opacity=1 ]   (311.92,68.27) -- (360.99,146.74) ;
\draw [shift={(362.05,148.44)}, rotate = 237.98] [color={rgb, 255:red, 255; green, 0; blue, 0 }  ,draw opacity=1 ][line width=0.75]    (10.93,-3.29) .. controls (6.95,-1.4) and (3.31,-0.3) .. (0,0) .. controls (3.31,0.3) and (6.95,1.4) .. (10.93,3.29)   ;
\draw [color={rgb, 255:red, 255; green, 0; blue, 0 }  ,draw opacity=1 ]   (379.15,164.85) .. controls (473.87,178.4) and (456.66,246.91) .. (431.96,292.61) ;
\draw [shift={(431.21,293.99)}, rotate = 298.77] [color={rgb, 255:red, 255; green, 0; blue, 0 }  ,draw opacity=1 ][line width=0.75]    (10.93,-3.29) .. controls (6.95,-1.4) and (3.31,-0.3) .. (0,0) .. controls (3.31,0.3) and (6.95,1.4) .. (10.93,3.29)   ;
\draw [color={rgb, 255:red, 255; green, 0; blue, 0 }  ,draw opacity=1 ]   (449.17,311.22) -- (498.01,232.62) ;
\draw [shift={(499.06,230.92)}, rotate = 481.85] [color={rgb, 255:red, 255; green, 0; blue, 0 }  ,draw opacity=1 ][line width=0.75]    (10.93,-3.29) .. controls (6.95,-1.4) and (3.31,-0.3) .. (0,0) .. controls (3.31,0.3) and (6.95,1.4) .. (10.93,3.29)   ;
\draw [color={rgb, 255:red, 255; green, 0; blue, 0 }  ,draw opacity=1 ]   (517.03,213.69) .. controls (613.53,183.46) and (589.9,119.84) .. (566.45,86.3) ;
\draw [shift={(565.38,84.79)}, rotate = 414.31] [color={rgb, 255:red, 255; green, 0; blue, 0 }  ,draw opacity=1 ][line width=0.75]    (10.93,-3.29) .. controls (6.95,-1.4) and (3.31,-0.3) .. (0,0) .. controls (3.31,0.3) and (6.95,1.4) .. (10.93,3.29)   ;
\draw [color={rgb, 255:red, 255; green, 0; blue, 0 }  ,draw opacity=1 ]   (565.38,50.32) .. controls (531.85,21.22) and (339.41,-19.82) .. (159.1,52.13) ;
\draw [shift={(159.1,52.13)}, rotate = 338.25] [color={rgb, 255:red, 255; green, 0; blue, 0 }  ,draw opacity=1 ][line width=0.75]    (10.93,-3.29) .. controls (6.95,-1.4) and (3.31,-0.3) .. (0,0) .. controls (3.31,0.3) and (6.95,1.4) .. (10.93,3.29)   ;
\draw   (211.23,217.12) .. controls (211.23,208.05) and (218.89,200.7) .. (228.34,200.7) .. controls (237.78,200.7) and (245.44,208.05) .. (245.44,217.12) .. controls (245.44,226.19) and (237.78,233.54) .. (228.34,233.54) .. controls (218.89,233.54) and (211.23,226.19) .. (211.23,217.12) -- cycle ;
\draw [color={rgb, 255:red, 0; green, 0; blue, 255 }  ,draw opacity=1 ][fill={rgb, 255:red, 0; green, 0; blue, 255 }  ,fill opacity=1 ]   (177.11,310.58) -- (275.68,311.21) ;
\draw [shift={(277.68,311.22)}, rotate = 180.36] [color={rgb, 255:red, 0; green, 0; blue, 255 }  ,draw opacity=1 ][line width=0.75]    (10.93,-3.29) .. controls (6.95,-1.4) and (3.31,-0.3) .. (0,0) .. controls (3.31,0.3) and (6.95,1.4) .. (10.93,3.29)   ;
\draw [color={rgb, 255:red, 0; green, 0; blue, 255 }  ,draw opacity=1 ]   (243.51,164.92) .. controls (338.23,178.46) and (321.02,246.98) .. (296.32,292.68) ;
\draw [shift={(295.57,294.05)}, rotate = 298.77] [color={rgb, 255:red, 0; green, 0; blue, 255 }  ,draw opacity=1 ][line width=0.75]    (10.93,-3.29) .. controls (6.95,-1.4) and (3.31,-0.3) .. (0,0) .. controls (3.31,0.3) and (6.95,1.4) .. (10.93,3.29)   ;
\draw [color={rgb, 255:red, 0; green, 0; blue, 255 }  ,draw opacity=1 ]   (228.34,200.7) .. controls (340.83,196.66) and (318.45,120.93) .. (295.03,87.03) ;
\draw [shift={(293.96,85.51)}, rotate = 414.31] [color={rgb, 255:red, 0; green, 0; blue, 255 }  ,draw opacity=1 ][line width=0.75]    (10.93,-3.29) .. controls (6.95,-1.4) and (3.31,-0.3) .. (0,0) .. controls (3.31,0.3) and (6.95,1.4) .. (10.93,3.29)   ;
\draw    (228.34,233.54) -- (277.68,311.22) ;
\draw    (177.11,310.58) -- (228.34,233.54) ;
\draw [color={rgb, 255:red, 0; green, 0; blue, 255 }  ,draw opacity=1 ]   (566.92,328.46) .. controls (537.7,383.45) and (395.01,405.39) .. (333.2,363.7) ;
\draw [shift={(332.28,363.07)}, rotate = 394.76] [color={rgb, 255:red, 0; green, 0; blue, 255 }  ,draw opacity=1 ][line width=0.75]    (10.93,-3.29) .. controls (6.95,-1.4) and (3.31,-0.3) .. (0,0) .. controls (3.31,0.3) and (6.95,1.4) .. (10.93,3.29)   ;

\draw (281.21,60.36) node [anchor=north west][inner sep=0.75pt]    {$a_{2}{}_{,2}$};
\draw (145.36,53.99) node [anchor=north west][inner sep=0.75pt]    {$a_{1}{}_{,1}$};
\draw (421.2,57.36) node [anchor=north west][inner sep=0.75pt]    {$a_{3}{}_{,1}$};
\draw (549.35,53.7) node [anchor=north west][inner sep=0.75pt]    {$a_{4}{}_{,}{}_{1}$};
\draw (149.21,299.85) node [anchor=north west][inner sep=0.75pt]    {$a_{1}{}_{,2}$};
\draw (280.21,297.22) node [anchor=north west][inner sep=0.75pt]    {$a_{2}{}_{,2}$};
\draw (417.63,298.22) node [anchor=north west][inner sep=0.75pt]    {$a_{3}{}_{,2}$};
\draw (214.51,151.74) node [anchor=north west][inner sep=0.75pt]    {$b_{1}{}_{,1}$};
\draw (349.93,153.11) node [anchor=north west][inner sep=0.75pt]    {$b_{2}{}_{,1}$};
\draw (486.78,154.48) node [anchor=north west][inner sep=0.75pt]    {$b_{3}{}_{,1}$};
\draw (351.35,202.36) node [anchor=north west][inner sep=0.75pt]    {$b_{2}{}_{,2}$};
\draw (485.35,203.73) node [anchor=north west][inner sep=0.75pt]    {$b_{3}{}_{,2}$};
\draw (554.35,298.56) node [anchor=north west][inner sep=0.75pt]    {$b_{4}{}_{,}{}_{1}$};
\draw (310.53,354.65) node [anchor=north west][inner sep=0.75pt]    {$r$};
\draw (344.82,73.44) node [anchor=north west][inner sep=0.75pt]  [rotate=-59.58]  {$x_{2} =1$};
\draw (482.83,80.06) node [anchor=north west][inner sep=0.75pt]  [rotate=-62.33]  {$x_{3} =1$};
\draw (331.47,284.54) node [anchor=north west][inner sep=0.75pt]  [rotate=-304.79]  {$x_{2} =1$};
\draw (465.89,293.3) node [anchor=north west][inner sep=0.75pt]  [rotate=-303.17]  {$x_{3} =1$};
\draw (288.96,197.7) node [anchor=north west][inner sep=0.75pt]  [rotate=-64.47]  {$y_{1} =1$};
\draw (397.22,131.48) node [anchor=north west][inner sep=0.75pt]  [rotate=-28.38]  {$y_{2} =1$};
\draw (398.88,213.77) node [anchor=north west][inner sep=0.75pt]  [rotate=-325.25]  {$y_{2} =1$};
\draw (592.86,161.87) node [anchor=north west][inner sep=0.75pt]  [rotate=-284.7]  {$y_{3} =1$};
\draw (526.24,209.22) node [anchor=north west][inner sep=0.75pt]  [rotate=-339.22]  {$y_{3} =1$};
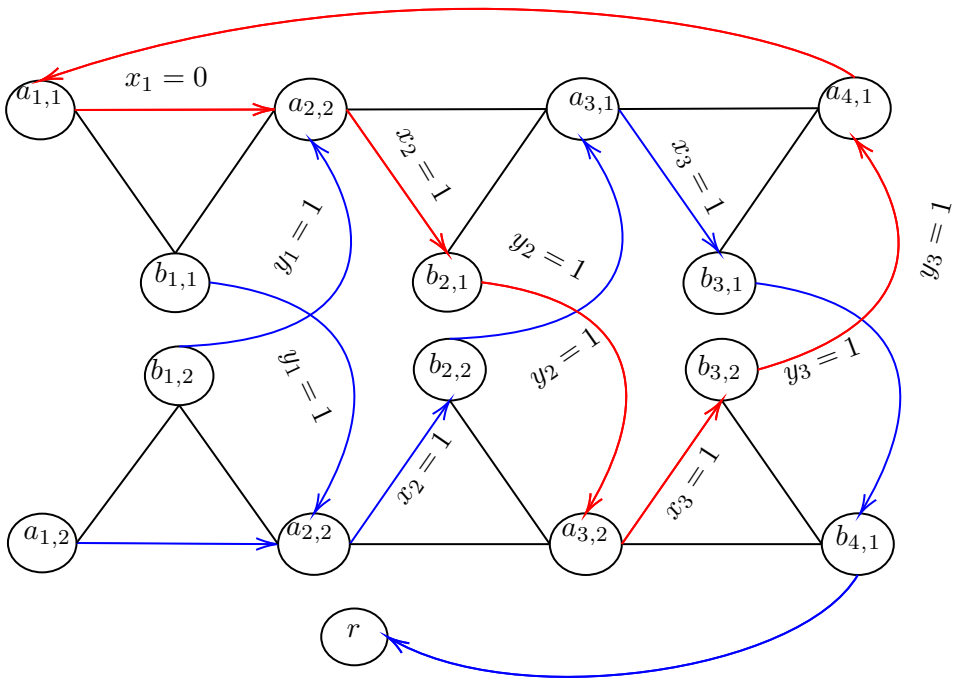
\draw (193.04,405.1) node [anchor=north west][inner sep=0.75pt]    {$Figure\ \ 4: \text{ Gadget for}\ \sum _{i=1}^{3} x_{i} y_{i} =0\bmod 2.$};
\draw (199.66,43.12) node [anchor=north west][inner sep=0.75pt]    {$x_{1} =0$};
\draw (212.51,205.74) node [anchor=north west][inner sep=0.75pt]    {$b_{1}{}_{,2}$};
\draw (270.1,155.55) node [anchor=north west][inner sep=0.75pt]  [rotate=-300.01]  {$y_{1} =1$};

\end{tikzpicture}
\newline
Now we give the proof of the claim in details below.
\begin{proof}
If $\angle{x,y}=0 \mod 2$, there is a path $\mathcal{P}_1 : a_{1,1} \rightsquigarrow a_{n+1,1}$ which flips layers even number of times and then completes a cycle via the edge $a_{n+1,1}\rightarrow a_{1,1}$. Thus in such a case there does not exist a path from $a_{1,1}$ to $r$, ruling out a spanning tree.  

On the other hand, if $\angle{x,y}=1 \mod 2$, there is a path $\mathcal{P}_2 : a_{1,1}\rightsquigarrow b_{n+1,1} \rightarrow r$ which flips layers an odd number of times. 
For every node $u \in G_{x,y}$ shown in the figure above that does not lie on the path $\mathcal{P}_2$, observe that either there exists a path $\mathcal{P}_u : u \rightsquigarrow a_{n+1,1}\rightarrow a_{1,1}\rightsquigarrow r$ or there exists a path $P'_u: u \rightsquigarrow b_{n+1,1}\rightarrow r$. In either case, therefore, there exists a path from $u$ to $r$. It is simple to  verify that no cycles are formed proving the existence of a spanning tree rooted at $r$.


\end{proof}

Now we are ready to prove  Lemma \ref{lem:fixed-dist}.

\begin{proof}
Define $\mathcal{R}=\{(\tau_{x}, \theta_{y}) : x,y\in\{0,1\}^n\times \{0,1\}^n\}\subseteq \Tau\times \Theta$. 
Observe that, the reduction $\red$ forms a bijection between $\{0,1\}^n \times \{0,1\}^n$ and $\mathcal{R}$. Thus a uniform distribution $U_{n}$ on $\{0,1\}^n \times\{0,1\}^n$ induces a uniform distribution on $\mathcal{R}$ via $\red$. We call the latter distribution $\mathcal{D}_P$.\\
Consider a rectangle $S=(S_1\times S_2) \subseteq \Tau\times \Theta$. We want to prove that $\disp_{\mathcal{D}_P}(S)= 2^{-\Omega(n)}$. Let $S'=S'_1\times S'_2$ be a sub rectangle of $S$ such that $S'_1=\{ \tau_x| \tau_x\in S_1, x\in \{0,1\}^n\}$ and $S'_2=\{\theta_y| \theta_y \in S_2, y \in \{0,1\}^n\}$.
Observe that $\disp_{\mathcal{D}_P}(S')=\disp_{\mathcal{D}_P}(S)$. Since $\red$ is a bijection between input instances of $\IP_n$ and $\mathcal{R}$, there exists a rectangle $T\subseteq \{0,1\}^n\times\{0,1\}^n$, such that $\red(T)=S'$. So, $\disp_{\mathcal{D}_P}(S')=\disp_{U_{n}}(T)$. Now using Theorem \ref{thm:ip-discrepancy} we conclude $\disp_{\mathcal{D}_P}(S)= 2^{-\Omega(n)}$. 



\end{proof}

\subsection{The Universal Distribution for Nearly-Balanced Partitions}\label{universal}

Lemma \ref{lem:fixed-dist} in the earlier section shows that for any \emph{fixed} perfectly balanced partition $P$, the function $g_P$ has low discrepancy under the distribution $\mathcal{D}_P$. 

However, to use these ideas in the context of lower bounds for monotone circuits, we need to handle a more subtle situation. Recall from Theorem \ref{thm:monotone-structure} that every product polynomial in the decomposition of the circuit, has its \emph{own} nearly-balanced partition of the first indices of the variables. This forces us to switch quantifiers, and specify a universal distribution that gives rise to small discrepancy for every nearly-balanced partition. 

We now consider a graph with a fixed set $W$ of vertices with 
$|W|=n'=c\cdot n$ where $c$ is an appropriate constant to be fixed later in the analysis. The universal distribution $\Delta$ is supported on a subset of the set $\Gamma=\{\nu |\nu : W\backslash\{1\}\rightarrow W \}$. 


Now we define the universal distribution $\Delta$ on an appropriate subset of $\Gamma$ by the following sampling process. 


\subsubsection*{Sampling process for $\Delta$: }
\begin{enumerate}
    \item Sample a random permutation $\pi : W\rightarrow W $ such that $\pi(1)=1$

    \item  Sample $(x,y)\in\{0,1\}^n\times \{0,1\}^n$ uniformly at random. 
    Construct ${G}_{x,y}$ as defined in the previous section \ref{gadget reduction} by applying the reduction $\red$ on $(x,y)$. 
    
       \item Embed an isomorphic copy of $G_{x,y}$ on the set of vertices $\pi(1),\ldots \pi(4n+3)$ using the following vertex map $E:V(G_{x,y})\to W$
    \begin{enumerate}
    \item For $0\leq j\leq n$ set $E(a_{j+1,1})=\pi(4j+2),\  E(a_{j+1,2})=\pi(4j+3),\  E(b_{j+1,1})=\pi(4j+4)$ and $ E(b_{j+1,2})=\pi(4j+5)$.
    \item Set $E(a_{n+1,1})=\pi(4n+2),\   E(b_{n+1,1})=\pi(4n+3),\ E(r)=\pi(1)$.
\end{enumerate}
\item Connect the remaining vertices of $W$ that is $\pi(4n+4),\ldots \pi(n')$ to $E(r)=\pi(1)$.
    
    \end{enumerate}

Let the vertices  of $W$ be partitioned into nearly-balanced sets $W_{A}\cup W_{B}$, where $W_A$ ($W_B$) is with Alice (Bob). Given a map from $\Gamma$, it induces two functions $\tau, \theta : W_A, W_B \rightarrow W$ which will be given to Alice and Bob respectively. 

The main result that we show is the following lemma which is the more precise re-statement of Lemma \ref{lem:variable-dist-intro}.

\begin{lemma}\label{lem:variable-dist}
Under the distribution $\Delta\sim \Gamma$, the discrepancy of $g_P$ is $2^{-\Omega(n)}$ for every nearly-balanced partition $P$. More precisely, 
$\disp_{\Delta}(C^{\ST_{n'}})$ is at most $2^{-\Omega_c(n)}$, where $n'=c n$ for a constant $c$. 
\end{lemma}

\subsection{Discrepancy Under the Universal Distribution}\label{low-discrepancy}
We need a notion of \emph{honouring} a sub-gadget under a permutation $\pi$. 

\begin{definition}\label{honour}

In ${G}_{x,y}$ each sub-gadget ${G}_i (1\leq i\leq n)$ has two pairs of vertices  $\{a_{i,1}, a_{i,2}\}$ and $\{b_{i,1}, b_{i,2}\}$. For any given nearly-balanced partition $P$ of $W$, we say a permutation $\pi$ honours the sub-gadget ${G}_i$ if for every $j$, $E(a_{i,j})=\pi(k)\in W_A$ and $E(b_{i,j})=\pi(l)\in W_B$ for some $k,l \in [n']\backslash \{1\}$. 
\end{definition}

For a suitable $\alpha\in(0,1)$ which we fix later, 
we say that
$\pi$ is good for a partition $P$ if $\pi$ honours at least $\alpha n$ sub-gadgets 
$G_i$ for $1\leq i\leq n$. Otherwise, we say that $\pi$ is bad. We prove a simple lemma showing that the random map $\pi$ is bad with low probability.

\begin{lemma}\label{bad-prob}
Let $P$ be any nearly  balanced partition. Then, for suitably chosen constants $\alpha$ and $c$,  $$\Pr[\pi ~\text{is bad}]\leq 2^{-\Omega(n)}.$$
\end{lemma}
\begin{proof}
Let $P=W_A\cup W_B$ and $|W_A|=k, |W_B|=t$. Let $\pi$ be any random permutation such that $\pi(1)=1$.
We give a simple lower bound on the probability of $\pi$ honouring sub-gadget $G_i$ conditioned on the vertices of the first $i-1$ sub-gadgets being mapped to an arbitrary set of vertices in $W$. More precisely, for $1 \le j < i$ and $2\le r \le 5$, let $w_{j,r} \in W$ be arbitrary vertices. Then,
\begin{equation*}
\begin{split}
\Pr\bigg[\pi ~\text{honours} ~{G}_i\, \bigg|\, \pi\big(4(j-1)+r\big) = w_{j,r}, 1\le j < i, 2 \le r \le 5  \bigg] &\geq\left(\frac{\binom{k-4(i-1)}{2}\binom{t-4(i-1)}{2}}{\binom{n'-1-4(i-1)}{4}}\right)
\end{split}
\end{equation*}
 Using standard upper and lower bound estimates for binomial coefficients, and the  fact that $k,t\geq n'/3$, we can easily check that the RHS above is at least some \emph{constant} $\beta >0$ for an appropriate choice of $c$. 
  

Thus, for any choice of $i_1,\ldots,i_{(1-\alpha)n} \in [n]$, 
$$ \Pr\bigg[\bigwedge_{\ell=1}^{\ell=(1-\alpha)n} \pi \text{ dishonours } G_{i_{\ell}}\bigg] \leq \big(1-\beta\big)^{(1-\alpha)n}$$

Using the union bound, the probability that 
 $\pi$ dishonours at least $(1-\alpha)n$ sub-gadgets is at most 
 
 $$\binom{n}{(1-\alpha)n}
 (1-\beta)^{(1-\alpha)n}.$$
 
Write $(1-\beta)=2^{-\delta}$ for some $\delta>0$. 
Applying the standard estimate $\binom{n}{k}=2^{H(\frac{k}{n})n}$, 
\[
\binom{n}{(1-\alpha)n}(1-\beta)^{(1-\alpha)n}
\leq 2^{n H(1-\alpha)} \cdot 2^{-\delta(1-\alpha)n}.
\]
By fixing a suitable $\alpha$ such that $H(1-\alpha)<\delta(1-\alpha)$, the estimate can be upper bounded by $2^{-\Omega(n)}$.  
\end{proof}

Now we are ready to prove Lemma \ref{lem:variable-dist}.

\begin{proof}
Let $H_P$ and $F_P$ be respectively the sets of good and bad permutations $\pi$ with respect to a partition $P$. 
Notice that, 
$$\disp_{\Delta}(g_P)\leq \Pr[\pi\in F_P ] + \sum_{\pi\in H_P} \Pr[\pi]\cdot \disp_{\Delta|_{\pi}}(g_P).$$ By Lemma \ref{bad-prob}, the first term is bounded by $2^{-\Omega(n)}$. We now show how to bound the second term $\disp_{\Delta|_{\pi}}(g_P)$ for any 
$\pi\in H_P$. Let $\Bad_{\pi}$ ($\Good_{\pi}$) be the set of indices $i$ such that the sub-gadgets ${G}_i$ are not honoured (honoured) by $\pi$. Now for any arbitrary partition $P$ of $W$ into $W_A, W_B$ define the set $\Tau_P=\{\psi|_{W_A}\ \bigg| \psi \in \supp({\Delta}),\ \psi|_{W_A} : W_A\rightarrow W \}$ ($\Theta_P=\{\psi|_{W_B}\ \bigg| \psi \in \supp({\Delta}),\ \psi|_{W_B} : W_B\rightarrow W \}$). 

When we sample a map according to distribution $\Delta$, the induced maps from $\Tau_P$ ($\Theta_P$) may depend on both $x,y\in \{0,1\}^n$ for an arbitrary partition $P$. Hence we denote these random induced maps by $\tau_{(x,y)}$, $\theta_{(x,y)}$. However the following is simple to verify.

\begin{remark}\label{good}
For a vertex $u \in W_A$ ($v\in W_B$) if $\pi^{-1}(u)= E(a_{i,j})$ ($\pi^{-1}(v)=E(b_{i,j})$) for some $i,j$,  $\tau(u)$ ($\theta(v)$) only depends on $x_i$ ($y_i$).
\end{remark}

For a partition $P$, define $Z_P$ to be a function from $\Tau_P\times \Theta_P \rightarrow \{1,-1\}$ such that 
$$Z_P(\tau,\theta)=
\begin{cases}
	1&\text{when}\ g_P(\tau,\theta)=0\\
    -1&\text{when}\ g_P(\tau,\theta)=1
\end{cases}
$$
For $\pi \in H_P$ we bound the following value, $\disp_{\Delta|_{\pi}}(Z_P,S_P)$ where $S_P$ is an arbitrary rectangle in $\Tau_P\times\Theta_P$. Here for sake of simplicity we abuse the notation $S_P$ and denote it as a characteristic function for the rectangle $S_P$. 
\begin{equation}
\begin{split}
    \disp_{\Delta|_{\pi}}(Z_P, S_P)=\bigg| \E_{(x,y)\sim {U_{n}} }\bigg[Z_P(\tau_{(x,y)},\theta_{(x,y)})\cdot S_P(\tau_{(x,y)}, \theta_{(x,y)})\bigg]\bigg|\\
   = \bigg| \E_{(u_1,v_1)\sim{U_{|\text{\tiny{Bad}}_{\pi}|}}} \bigg[ \E_{(u_2,v_2)\sim_{U_{|\text{\tiny{Good}}_{\pi}|}}} \bigg[ Z_P(\tau_{(u_1\circ u_2, v_1\circ v_2)},\theta_{(u_1\circ u_2, v_1\circ v_2)})\cdot S_P(\tau_{(u_1\circ u_2, v_1\circ v_2)},\theta_{((u_1\circ u_2, v_1\circ v_2)})\bigg]\bigg]\bigg|
\end{split}
 \end{equation}
 For a good map $\pi \in H_P$, let $V_{\Good_{\pi}}$ be the set of vertices from the sub-gadgets $G_i$ that are honoured by $\pi$, i.e.
  $V_{\Good_{\pi}} = \bigcup_{i\in \Good_{\pi}} V(G_i)$. 
 
 Fixing any $u_1,v_1 \in \{0,1\}^{\Bad_{\pi}}$, we denote the restricted map  $\tau_{(u_1\circ u_2,v_1\circ v_2)}\in \Tau_P$ as \[\tau_{(u_2,v_2)}^{(u_1,v_1)} : W_A \cap E(V_{\Good_{\pi}})\to W.\] Using Remark~\ref{good} we observe that, for any vertex $p \in W_A\cap E(V_{\Good_{\pi}})$, $\tau_{(u_2,v_2)}^{(u_1,v_1)}(p)$ only depends on variables from $u_2$. Similarly after fixing $u_1,v_1 \in \{0,1\}^{\Bad_{\pi}}$ we denote the restricted map $\theta_{(u_1\circ u_2,v_1\circ v_2)}$ as  \[\theta_{(u_2,v_2)}^{(u_1,v_1)} : W_B \cap E(V_{\Good_{\pi}}) \to W.\] Observe that for any vertex $q \in W_B\cap E(V_{\Good_{\pi}})$, $\theta_{(u_2,v_2)}^{(u_1,v_1)}(q)$ only depends on variables from $v_2$.\\
 Let $S_P^{(u_1,v_1)}$ be the rectangle in $\Tau_{P}^{(u_1,v_1)} \times \Theta_{P}^{(u_1,v_1)}$,  where $\Tau_P^{(u_1,v_1)}\coloneqq\{\tau_{(u_2,v_2)}^{(u_1,v_1)} \bigg|\  \tau_{(u_1\circ u_2,v_1\circ v_2)}\in \Tau_P\}$ and  $ \Theta_P^{(u_1,v_1)}\coloneqq\{\theta_{(u_2,v_2)}^{(u_1,v_1)} \bigg|\  \theta_{(u_1\circ u_2,v_1\circ v_2)}\in \Theta_P\}$. For simplicity, we denote maps from $\Tau_P^{(u_1,v_1)} $ and $\Theta_P^{(u_2,v_2)}$ as $\tau'_{(u_2)}, \theta'_{(v_2)}$. Now in Equation $(3)$ the inner expectation becomes following,
\[
\E_{(u_2,v_2)\sim{U_{|\text{\tiny{Good}}_{\pi}|}}} \bigg[ Z_P^{(u_1,v_1)}(\tau'_{(u_2)},\theta'_{(v_2)})\cdot S_P^{(u_1,v_1)}(\tau'_{(u_2)},\theta'_{(v_2)})  \bigg] 
\]
Say the above expression is $D_{u_1,v_1}$. Let us define a rectangle $R^{(u_1,v_1)}$ in the space $\{0,1\}^{\Good_{\pi}} \times \{0,1\}^{\Good_{\pi}}$ by setting $R^{(u_1,v_1)}(u_2,v_2)=  S_P^{(u_1,v_1)}(\tau'_{(u_2)},\theta'_{(v_2)}) $. 

The definition of the maps $\tau,\theta$ using $\red$ implies that $Z_P^{(u_1,v_1)}$ is either $\IP_{|\Good_{\pi}|}$ or $\overline{\IP_{|\Good_{\pi}|}}$.  Without loss of generality, assume it computes $\IP_{|\Good_{\pi}|}$. Then, 
\[D_{u_1,v_1}= \E_{(u_2,v_2)\sim U_{|\Good_{\pi}|}}\bigg[\IP_{|\Good_{\pi}|}(u_2,v_2)\cdot R^{(u_1,v_1)}(u_2,v_2)   \bigg].\]
Therefore, \[|D_{u_1,v_1}| \leq \disp_{U_{|\Good_{\pi}|}}(\IP_{|\Good_{\pi}|}(u_2,v_2)) \leq 2^{-\frac{\big|\Good_{\pi}\big|}{2}},\]where the third inequality follows from Theorem \ref{thm:ip-discrepancy}. 
Hence,
\begin{equation}
\begin{split}
\bigg| \E_{(u_1,v_1)\sim U_{|\text{\tiny{Bad}}_{\pi}|}}\bigg[ D_{u_1,v_1} \bigg] \bigg| \leq \E_{(u_1,v_1)\sim U_{|\text{\tiny{Bad}}_{\pi}|}}\bigg|\bigg[ D_{u_1,v_1} \bigg]\bigg| \leq \E_{(u_1,v_1)\sim U_{|\text{\tiny{Bad}}_{\pi}|}}\bigg[ 2^{\frac{-\big|\Good_{\pi}\big|}{2}} \bigg] \leq 2^{\frac{-\big|\Good_{\pi}\big|}{2}} 
\end{split}
\end{equation}

Since $|\Good_{\pi}|=\alpha n$, we conclude that $\disp_{\Delta|_{\pi}}(Z_P)=\disp_{\Delta|_{\pi}}(g_P)= 2^{-\Omega(n)}$.

\end{proof}

\section{$\epsilon$-Sensitive Lower Bound for Spanning Tree Polynomial}\label{sec:lb-spanning-tree-polynomial}
Let us recall Theorem \ref{thm:spanning-tree-sensitive} that we prove next.

\sensitivespanningtree*

\begin{proof}

First we recall the spanning tree polynomial for complete graph $G$ on $n$ vertices $\{1,2,\ldots,n\}$, where every monomial corresponds to a spanning tree in $G$ rooted at vertex $1$, i.e.
\[
\ST_n(\widetilde{G}) = \sum_{\nu \in T_n} x_{2,\nu(2)}x_{3,\nu(3)}\cdots x_{n,\nu(n)},
\]
where $T_n=\{\nu:\{2,3,\ldots,n\} \mapsto \{1,2,\ldots,n\} ~|~ \forall i ~\exists k \ \nu^{k}(i)=1 \ ; \forall i  \ (i,\nu(i)) \in E(\widetilde{G}) \}$ and recall from Section \ref{sec:intro} that $\widetilde{G}$ is the graph obtained from $G$ which has edges in both directions for each undirected edge in $G$. 
We want to invoke the Discrepancy-Sensitivity Correspondence theorem (Theorem \ref{thm-sensitivity-intro}) by using Lemma \ref{lem:variable-dist} to finish the proof of our Theorem \ref{thm:spanning-tree-sensitive}. Towards this effect set $f$ in Theorem \ref{thm-sensitivity-intro} as $\ST_n$. Using Lemma \ref{lem:variable-dist} we observe that $\disp_{\Delta}(C^f)=2^{-\eta_{0}\cdot n}$ for some constant $\eta_{0}>0$ where $\Delta$ is the universal distribution defined in Lemma \ref{lem:variable-dist}. This we denote as $\gamma$. It is simple to verify that choosing $\epsilon=2^{-\frac{\eta_{0}}{10}\cdot n}$ satisfies the condition in Theorem \ref{thm-sensitivity-intro}, i.e. $\epsilon \geq \frac{6\gamma}{1-3\gamma}$. \\
Now using Theorem \ref{thm-sensitivity-intro} we conclude that the monotone complexity of $F_{n-1,n} - \epsilon\cdot \ST_n$ is at least $\frac{\epsilon}{3\gamma}=2^{\Omega(n)}$. By analogous argument the monotone complexity of $F_{n-1,n} + \epsilon\cdot \ST_n$ is also at least $2^{\Omega(n)}$.  

\end{proof}

\section{Open Problems}\label{sec:openproblem}
Our work raises several questions for further investigation. We outline some of them below:

\begin{itemize}
\item It will be very interesting to find a family of explicit polynomials in $\VP$ which can be computed by polynomial-size \emph{formulas} but the monotone circuit complexity is strongly exponential in the number of variables. The known $\VP$ upper bound for the family of spanning tree polynomials only provides polynomial-size algebraic branching programs. Since the upper bound result uses Matrix Tree Theorem involving non-trivial determinant computation, it seems unlikely that in general a family of spanning tree polynomials can be computed by polynomial-size formulas. The known simulation of algebraic branching programs by formulas gives only quasi-polynomial formula size upper bound. One of the main results in \cite{CDM21} does exhibit a family of explicit polynomials computed by \emph{depth-three} arithmetic formulas whose monotone circuit complexity is exponential. However, the lower bound is not strongly exponential in the number of variables.     

\item Another concrete open problem is to improve the result of Theorem \ref{thm:spanning-tree-sensitive} quantitatively. Since the family of spanning tree polynomials $\ST_n$ defined over the complete graphs have $n^2$ variables, the $\epsilon$-sensitive lower bound is not strongly exponential in the number of variables. A natural goal is to obtain strongly exponential $\epsilon$-sensitive monotone lower bounds.
This will subsume both Theorems~\ref{thm:spanning-tree} and \ref{thm:spanning-tree-sensitive} of this work.

\item How sensitive is the following monotone polynomial : $F_{n,n} -\epsilon\cdot \det_{n,n}$, where $\det_{n,n}$ is the determinant polynomial for the $n\times n$ symbolic matrix ? Although, we are able to prove $\epsilon$-sensitive lower bound for the spanning tree polynomial (whose efficient computation involves determinant), we are unable to prove such a result for the determinant polynomial itself. To understand the difficulty for proving robust/sensitive lower bounds for the determinant, consider the following natural communication problem: Alice gets a map $\tau\,:\,\{1,\ldots,n/2\} \to [n]$ and Bob $\theta\,:\, \{n/2 + 1,\ldots,n\} \to [n]$. They have to decide if $\tau,\theta$ together form a permutation of $[n]$. It is simple to verify that this communication problem shows up when one tries to prove $\epsilon$-sensitive lower bounds for either the determinant or the permanent via our discrepancy-sensitivity correspondence. However, this communication problem has even an efficient bounded-error randomized protocol, ruling out the direct use of our current technique. Our work thus throws interesting challenge to  prove such robust monotone lower bounds for either the determinant or the permanent. 

\item Another very interesting question is to understand the usefulness of discrepancy based techniques or complexity measures against general (set-)multilinear circuits. Is there a polynomial efficiently computed by a (set-)multilinear circuit that has large $\epsilon$-sensitive monotone complexity? Spanning tree polynomials are not known have efficient (set-)multilinear circuits. The answer to this question seems difficult to guess. It is worth noting that until recently no super-polynomial separation between the powers of multilinear and ordinary monotone computations were known. Such a separation was just achieved recently by \cite{CDM21}. This separation used a \emph{corruption} based argument that is not known to yield $\epsilon$-sensitive lower bounds.
\end{itemize}

\section*{Acknowledgement} 
We thank Mrinal Kumar for his comments on an earlier version of this work.

\bibliographystyle{amsalpha} 
\bibliography{reference}
\newpage
\appendix

\end{document}